\documentclass[a4paper,USenglish,cleveref, autoref, thm-restate]{lipics-v2021}

\hideLIPIcs
\nolinenumbers

\title{Covering and Partitioning Complex Objects with Small Pieces} %TODO Please add

\author{Anders Aamand}{Department of Computer Science, University of Copenhagen, Denmark}{aa@di.ku.dk}{0000-0002-0402-0514}{This work was supported by the VILLUM Foundation grant 54451}
\author{Mikkel Abrahamsen}{Department of Computer Science, University of Copenhagen, Denmark}{miab@di.ku.dk}{https://orcid.org/0000-0003-2734-4690}{Supported by Independent Research Fund Denmark, grant 1054-00032B, and by the Carlsberg Foundation, grant CF24-1929.}
\author{Reilly Browne}{Department of Computer Science, Dartmouth College, Hanover, NH, USA}{reilly.browne.gr@dartmouth.edu}{https://orcid.org/0000-0003-3725-5245}{}
\author{Mayank Goswami}{Department of Computer Science, Queens College CUNY, New York}{mayank.goswami@qc.cuny.edu}{https://orcid.org/0000-0002-2111-3210}{Supported by National Science Foundation CCF-2503086}
\author{Prahlad Narasimhan Kasthurirangan}{Department of Applied Mathematics and Statistics, Stony Brook University, USA}{prahladnarasim.kasthurirangan@stonybrook.edu}{https://orcid.org/0000-0002-8518-7745}{}
\author{Linda Kleist}{Department of Informatics, Universtity Hamburg, Germany}{linda.kleist@uni-hamburg.de}{https://orcid.org/0000-0002-3786-916X}{}
\author{Joseph S. B. Mitchell}{Department of Applied Mathematics and Statistics, Stony Brook University, USA}{joseph.mitchell@stonybrook.edu}{[http://orcid.org/0000-0002-0152-2279]}{}
\author{Valentin Polishchuk}{Communications and Transport Systems, Linköping Univeristy, Sweden}{valentin.polishchuk@liu.se}{https://orcid.org/0000-0002-8292-2281}{}
\author{Jack Stade}{Department of Computer Science, University of Copenhagen, Denmark}{jast@di.ku.dk}{https://orcid.org/0009-0007-9153-6589}{Supported by Independent Research Fund Denmark, grant 1054-00032B, and by the Carlsberg Foundation, grant CF24-1929.}

\authorrunning{Aamand, Abrahamsen, Browne, Goswami, Kasthurirangan, Kleist, Mitchell, Polishchuk, Stade} 

\Copyright{} 



\Crefname{algo}{Algorithm}{Algorithms}
\Crefname{case}{Case}{Cases}
\Crefname{subsection}{Subsection}{Subsections}

\newcommand{\old}[1]{}

\newcommand{\eps}{\varepsilon}
\newcommand{\OPT}{\text{OPT}}

\newcommand{\poly}{\text{poly}}

\newcommand{\SmallCover}{\textsc{Small-Cover}}
\newcommand{\SmallPartition}{\textsc{Small-Partition}}
\newcommand{\DualSmallCover}{\textsc{Dual-Small-Cover}}
\newcommand{\SetCover}{\textsc{Set-Cover}} 

\begin{document}

\maketitle

\begin{abstract}
We study the problems of covering or partitioning a polygon $P$ (possibly with holes) using a minimum number of \emph{small pieces}, where a small piece is a connected sub-polygon contained in an axis-aligned unit square.
For covering, we seek to write $P$ as a union of small pieces, and in partitioning, we furthermore require the pieces to be pairwise interior-disjoint.
We show that these problems are in fact equivalent: Optimum covers and partitions have the same number of pieces.

For covering, a natural local search algorithm repeatedly attempts to replace $k$ pieces from a candidate cover with $k-1$ pieces.
In two dimensions and for sufficiently large $k$, we show that when no such swap is possible, the cover is a $1+\mathcal O(1/\sqrt k)$ approximation, hence obtaining the first PTAS for the problem.
Prior to our work, the only known algorithm was a $13$-approximation that only works for polygons without holes [Abrahamsen and Rasmussen, SODA 2025].
In contrast, in the three dimensional version of the problem, for a polyhedron $P$ of complexity $n$, we show that it is NP-hard to approximate an optimal cover or partition to within a factor that is logarithmic in $n$, even if $P$ is simple, i.e., has genus $0$ and no holes. 
\end{abstract}

\section{Introduction}
Partitioning large objects into pieces of bounded size is a fundamental problem.
One practical motivation comes from 3D printing and related manufacturing technologies: we can only produce objects that fit within the printing volume, which can be assumed to be a unit cube.
If the goal is to produce a large part $P$, we must partition $P$ into smaller pieces, $Q_1,\ldots,Q_k$, each of which fits within the unit cube.
The small pieces can then be assembled to form the desired part $P$.
A natural objective is to minimize the number, $k$, of pieces.
Heuristics have been proposed to solve this problem~\cite{DBLP:journals/tog/LuoBRM12,chen2022skeleton,jiang2017models,KO2025104759,YAO20241002,DBLP:journals/cgf/VanekGBMCSM14}.

Similar partitioning problems arise in everyday settings.
For instance, one might disassemble an entire broccoli head, the peeled stem included, into pieces suitable for consumption with minimal indignity.
Yet even in this humble task, one strives to avoid the disarray of producing an excessive multitude of fragments.
Recently, Abrahamsen and Rasmussen~\cite{DBLP:journals/corr/abs-2211-01359} initiated a theoretical study of partitioning polygons into few small pieces. They provided constant-factor approximation algorithms when the domain $P$ is a simple polygon (no holes), for various definitions of what it means for a piece $Q_i$ to be ``small,'' each with its own algorithm and approximation factor.
A particularly simple version demands that each piece~$Q_i$ be connected and contained in an axis-aligned unit square.
Abrahamsen and Stade~\cite{abrahamsen2026hardness} showed that this version is NP-hard even when $P$ is a simple polygon.

In this paper, we consider the problems of partitioning and covering objects with small pieces in two and three dimensions. We consider the setting where a small piece is defined as a connected subset of $P$ contained in an axis-aligned unit square (in two dimensions) or unit cube (in three dimensions). We obtain the following results.

\begin{itemize}
\item We present a PTAS for covering a polygon (possibly with holes) into a minimum number of small pieces (\Cref{sec:ptas:smallcover}).

\item For polygons (possibly with holes), we show that the small-cover and small-partition problem have the same optimal value (\Cref{sec:equivalence}).

\item We prove that in 3D, covering or partitioning a polyhedron $P$ {of description complexity $n$} 
with small pieces is NP-hard {to approximate  within} a logarithmic factor of $n$, even when $P$ is simply-connected and contained in a box of size $1\times 2\times 3$ (\Cref{sec:hardness3D}).

\end{itemize}

Our PTAS works by local search: we start with a candidate small cover and then repeatedly try to find a set of $k$ pieces that can be swapped for a set of $k-1$ pieces.
We show that this gives a $1+\mathcal O(1/\sqrt k)$ approximation algorithm, when $k$ is sufficiently large.
A result by Roy, Govindarajan, Raman, and Ray \cite{NonpiercingPTAS} shows that local search is a PTAS for discrete \SetCover\ when the candidate pieces satisfy a certain \emph{non-piercing} property. 
Two simply-connected pieces $P,Q$ are \emph{non-piercing} if $P\setminus Q$ and $Q\setminus P$ are connected.
We show that a set {of} maximal small pieces satisfies a slightly weaker non-piercing property, but we are able to modify the pieces slightly so that the result from \cite{NonpiercingPTAS} can be used. 
We are then able to analyze our ``continuous'' local search algorithm using this result{--unlike the discrete case studied in \cite{NonpiercingPTAS}, note that our search space is infinite}. 

\begin{figure}[ht]
\centering
\includegraphics[page=12,width=0.2\textwidth]{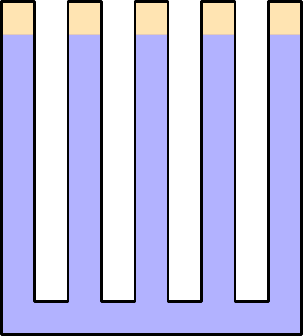}
\caption{A comb-shaped polygon might require an arbitrarily large number of pieces for a small cover, even if it fits in a box of size $1\times (1+\eps)$.}
\label{fig:comb}
\end{figure}

For many optimization algorithms, local search can be used to obtain a PTAS when a certain related graph has sublinear separators (see \cite{SeparatorLocalSearch} for details). We highlight the fact that this is non-obvious for the small-cover problem.
Local search gives a PTAS for \emph{packing} axis-aligned unit squares into a polyhedron in any {constant} dimension, since a density argument can be used to obtain the required separation bounds (see for example \cite{DensitySeparators}).
In contrast, for small covers, we insist that the pieces be connected, which means that there are ``comb-shaped'' polygons requiring many pieces, even though the polygon itself almost fits in a unit square; see \Cref{fig:comb}.
In 3D, we show that it is NP-hard to approximate optimal small cover or partition to within a factor of $(1-\alpha)\ln(n)$, for any $\alpha>0$, where $n$ is the number of bits needed to describe the input polyhedron.
We also observe (Theorem~\ref{thm:=}) that for small pieces, covering and partitioning are equivalent, in the sense that the smallest number of pieces in a cover and in a partition is the same. This is in stark contrast with other geometric settings where not only the optima differ, but also covering seems considerably harder than partitioning.
For instance, it is $\exists\mathbb R$-complete to cover a simple polygon with the minimum number of star polygons \cite{DBLP:journals/jacm/AbrahamsenAM22} (that is, the art-gallery problem) or convex polygons \cite{DBLP:conf/focs/Abrahamsen21}, while the corresponding partitioning problems are solvable in polynomial time~\cite{DBLP:conf/stoc/AbrahamsenBNZ24,chazelle1985optimal}.

\subsection{Related Work}

The small cover and partitioning problems belong to the class of \emph{decomposition problems}, which form a large, established subfield in computational geometry.
In all of these problems, we want to \emph{decompose} a polygon $P$ into connected polygonal \emph{pieces} that must respect certain restrictions.
Here, the union of the pieces must be $P$, and we usually seek a decomposition into as few pieces as possible, in which case the decomposition is called \emph{optimal}.
A decomposition that allows the pieces to overlap is called a \emph{cover}; if thee pieces are pairwise interior-disjoint, it is called a \emph{partition}.
Depending on the assumptions about the input polygon $P$ and the requirements on the pieces, this leads to a wealth of interesting problems.
There is a vast literature about such decomposition problems, as documented in several books and survey papers that give an overview of the state-of-the-art at {their respective times} of publication~\cite{shermer1992recent, chazelle1994decomposition, chazelle1985approximation, keil1999polygon, keil1985minimum, o1987art,o2004polygons}.

Earlier works studied variants of the small partitioning problem without Steiner points~\cite{damian2004computing,buchin2021decomposing}, {i.e., each vertex of a piece $Q_i$ is a vertex of $P$.}
{As an} undesirable consequence such a partition often does not exist, e.g., if an edge of the polygon $P$ is too long to fit in a piece.

Motivated by indoor localization using lasers, Arkin, Das, Gao, Goswami, Mitchell, Polishchuk, and T{\'{o}}th~\cite{DBLP:conf/esa/ArkinD0GMPT20} studied partitioning a polygon with chords (maximal segments within $P$), which may model laser beams that serve as ``tripwires''; the goal is to cut the polygon into {bounded} pieces using {the} fewest lasers.
Motivated by
communication graphs for mobile robots, Fekete, Kamphans, Kr{\"{o}}ller, Mitchell, and Schmidt~\cite{fekete2011exploring} (see also the video~\cite{DBLP:conf/compgeom/BeckerFKLMS13}) showed how to compute a Steiner triangulation of a simple polygon, whose edge lengths are bounded; they presented a $3$-approximation for minimizing the number of Steiner points.
Worman~\cite{worman2003decomposing} showed NP-hardness of partitioning a polygon with holes into {a minimum number of pieces contained in} axis-aligned squares, without Steiner points. More recently, Buchin and Selbach~\cite{buchin2021decomposing} proved NP-hardness of partitioning and covering with fat pieces. 
Combinatorial aspects of partitioning into small[er] pieces are classical problems in mathematics.
For instance, \emph{Borsuk's problem} is to partition a unit-diameter convex body (in any dimension) into fewest pieces of diameter strictly less than~$1$; see~\cite{borsuk1933drei,DBLP:journals/combinatorics/JenrichB14,kahn1993counterexample} for the rich history of the problem.
\emph{Conway's fried potato problem} \cite{bezdek1995solution, bezdek1996conway,canete2022conway,croft2012unsolved} deals with minimizing the in-radius of the pieces obtained by a specified number of hyperplane cuts of a given convex polyhedron (in $d$ dimensions); in our terminology, Conway's problem is a \textit{dual} small partitioning problem.

\subsection{Problem Formulation}

We say that a polygon $Q$ (possibly with holes) is \emph{small} if $Q$ is contained in an axis-aligned unit square.
We often {call} small polygons small \emph{pieces}.
A set $\mathcal Q$ of small pieces is a \emph{small cover} of a polygon $P$ if $\bigcup \mathcal Q=P$; 
note in particular that $Q\subseteq P$ for each small piece $Q\in\mathcal Q$.
We say that a piece $Q\in\mathcal Q$ is \emph{complete} if $Q$ is the closure of a connected component of the intersection $P^\circ\cap S$, where $P^\circ$ is the interior of $P$ and $S$ is an axis-aligned unit square.
A simpler definition would be that $Q$ is complete if it is a connected component of $P\cap S$, but that could be a degenerate polygon, which we by definition don't allow. 
A small cover $\mathcal Q$ of $P$ is a \emph{small partition} of $P$ if furthermore the pieces in $\mathcal Q$ are pairwise interior-disjoint.
In this paper, we study the problems of finding small covers and small partitions consisting of a minimum number of pieces for a given polygon $P$, which may have holes.
We denote these problems \SmallCover\  and \SmallPartition, respectively.

\section{A PTAS for \SmallCover\ in 2 dimensions}\label{sec:ptas:smallcover}

In this section, we present an 
$(1+\eps)$-approximation algorithm for \SmallCover\ when the input is a polygon $P$, possibly with holes.
In \cref{{sec:localSearch}}, we establish tools that allow us to use local search.
We then show in \Cref{sec:smalldiameter} how to use these techniques to handle polygons which intersect at most $D$ squares in the {integer} grid.
The running time of our algorithm {are} polynomial in $D$.
Finally, in \Cref{sec:ptas:largepolygons}, we extend our ideas to handle polygons that might have exponentially large diameter. {We note that, for the purposes of constructing approximate solutions for \SmallCover, we will only consider \emph{complete} small pieces, since any small piece that fits in a square $S$ must cover a subset of the complete piece covering the same component of $P^\circ \cap S$.} 

\subsection{Fundamentals for Local Search}\label{sec:localSearch}

Roy, Govindarajan, Raman, and Ray~\cite{NonpiercingPTAS} studied the geometric set cover problem with so-called non-piercing regions.
Two compact simply-connected regions $A$ and $B$ are \emph{non-piercing} if $A\setminus B$ and $B\setminus A$ are both connected.
We say that $A$ and $B$ intersect \emph{properly} if the intersection of the boundaries $\partial A\cap\partial B$ consists of a finite set of points and in each of these points, the boundaries cross, i.e., tangential intersections are not allowed.
We say that a set of compact, simply-connected regions is \emph{properly non-piercing} if the regions are pairwise non-piercing and intersect properly.

\begin{lemma}[\cite{NonpiercingPTAS}]\label{lem:localsearchpoints}
For a universal constant $C$ the following holds:
Let $X$ be a (finite) set of points in $\mathbb{R}^2$ and suppose that $\mathcal{Q}=\{Q_1, \dots, Q_q\}$ and $\mathcal{R}=\{R_1, \dots, R_r\}$ are sets of connected regions such that $\mathcal{Q}$ and $\mathcal{R}$ each cover $X$.
If the regions in $\mathcal{Q}\cup \mathcal{R}$ are properly non-piercing and $r\ge (1+\eps)q$ for some $1\ge \eps>0$, then there is some set of at most $C/\eps^2$ of the pieces in $\mathcal{R}$ that can be replaced by a strictly smaller set of pieces in $\mathcal{Q}$ to obtain a cover of $X$.
\end{lemma}

We start by adapting this result to the setting of covering a polygon rather than a discrete set of points.
We show the following:

\begin{lemma}\label{lem:localsearchregions}
There is a constant $C$ such that the following holds:
Let $P$ be a polygon (possibly with holes) and let $\mathcal{Q}$ be an optimal small cover of $P$ consisting of $q$ complete pieces, and let $\mathcal{R}$ be another small cover of $P$ consisting of $r$ complete pieces.
Then for any $1\ge\eps>0$, if $r\ge (1+\eps)q$, then there is a set of at most $C/\eps^2$ of the pieces in $\mathcal{R}$ that can be replaced by a strictly smaller set of pieces in $\mathcal{Q}$ to obtain a small cover of $P$.
\end{lemma}

It is relatively straightforward to prove \Cref{lem:localsearchregions} given \Cref{lem:localsearchpoints}.
The strategy is to
put a point in the interior of each cell of the overlay of $\mathcal Q $ and $\mathcal R$.
A subset of $\mathcal{Q}\cup \mathcal{R}$ covers $P$ if and only if it covers each of these points.
The pieces in $\mathcal{Q}\cup \mathcal{R}$ are almost non-piercing, but need to be modified slightly in order to remove degeneracies and holes.
The process is demonstrated informally in \Cref{fig:pieces_pierce} and uses an idea from~\cite{DBLP:journals/algorithmica/BergKMT23}.
The following lemma says that these modifications exist:

\begin{figure}[ht]
    \centering
    \includegraphics[page=1]{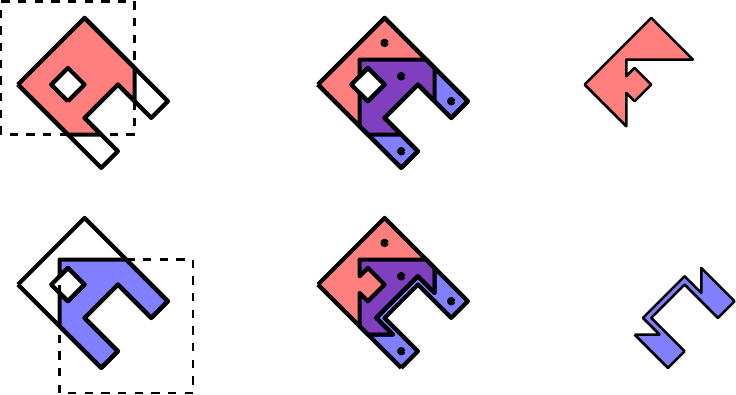}
    \caption{Given a finite set of complete small pieces that cover $P$, we obtain an equivalent instance of non-piercing geometric set cover. Left: a polygon and two complete small pieces. Top middle: the polygon and pieces cut the polygon into $4$ cells, and we place a point in each cell. These pieces are piercing, since removing the red piece from the blue piece would disconnect it. Additionally, the red piece has a hole. Bottom middle: modifying the pieces to remove the hole and make them non-piercing. The hole in the red piece is filled, and an interval of the boundary is offset inwards so as to not pierce the blue piece. Right: removing the modified blue piece from the red piece (top) or the modified red piece from the blue piece (bottom). The remaining parts are connected, so the modified pieces are non-piercing.}
    \label{fig:pieces_pierce}
\end{figure}

\begin{lemma}\label{lem:nondegenerate}
Consider a polygon $P$ and a small cover $\mathcal Q=\{Q_1,\ldots,Q_q\}$  of $P$ consisting of complete small pieces.
Consider the overlay of the pieces in $\mathcal Q$, choose a point in the interior of each region contained in $P$, and let $X$ be the set of these points.
There exists a set $\mathcal Q'=\{Q'_1,\ldots,Q'_q\}$ of simply-connected pieces that are properly non-piercing such that $Q_i\cap X=Q'_i\cap X$ for all $i\in \{1,\ldots,{q}\}$.
It holds in particular that for any $I\subset\{1,\ldots,{q}\}$, we have $X\subset \bigcup_{i\in I} Q'_i$ if and only if $P=\bigcup_{i\in I} Q_i$.
\end{lemma}

\Cref{lem:nondegenerate} {is proved in \cref{app:A}}. 
Note that these modifications are only used in the analysis; we don't need to produce them algorithmically.
{We now prove} \Cref{lem:localsearchregions}.

\begin{proof}[Proof of \Cref{lem:localsearchregions}]
Let $\mathcal{S}=\mathcal{Q}\cup \mathcal{R}$ be the union of the global optimum and the local optimum. Using \Cref{lem:nondegenerate}, we obtain a new set $\mathcal{S}'=\mathcal{Q}'\cup\mathcal{R}'$ of properly non-piercing pieces and a finite set of points $X$ such that a subset of $\mathcal{S}'$ covers $X$ if and only if the corresponding subset of $\mathcal{S}$ covers $P$. Since $\mathcal{Q}$ and $\mathcal{R}$ both cover $P$, $\mathcal{Q}'$ and $\mathcal{R}'$ must both cover $X$.

Recall that $q=|\mathcal{Q}|=|\mathcal{Q}'|$ and $r=|\mathcal{R}|=|\mathcal{R}'|$.
If $r\ge (1+\eps)q$ for some $1\ge \eps>0$, then by {\Cref{lem:localsearchpoints}}, we can replace a set of at most $C\frac1{\eps^2}$ of the pieces in $\mathcal{R}'$ with a strictly smaller set of pieces from $\mathcal{Q}'$ to obtain another cover of $X$.
If we swap the corresponding pieces in $\mathcal{R}$ with the corresponding pieces in $\mathcal{Q}$, then we get a small cover of $P$, proving the claim. 
\end{proof}

\Cref{lem:localsearchregions} gives sufficient conditions for a small cover to be a $(1+\eps)$-approximation of the optimum. We say that a small cover $R$ is a \emph{$k$-local optimum} if no set {of} at most $k$ pieces from $\mathcal R$ can be replaced by a strictly smaller set of small pieces.

\begin{corollary}\label{cor:localsearchcontinuous}
There is some constant $C'$ such that if $P$ is a polygon and $\mathcal{R}$ is a small cover that is a $k$-local optimum for $k\geq C'$ and has $r$ complete pieces, then $r$ is at most $1+\frac{C'}{\sqrt{k}}$ times the size of an optimal cover of $P$.
\end{corollary}

\begin{proof}
Taking $C$ to be the constant from \Cref{lem:localsearchregions}, let $C'=\sqrt{C}$ and $k_0=\lceil C'\rceil$.
Because $\mathcal R$ is $k$-local optimum, no set of at most $k$ pieces can be replaced by a smaller set of small pieces.
If $r\le k$, then $\mathcal{R}$ is necessarily an optimal solution and the claim holds. So suppose that $r>k$.
Let $\mathcal{Q}$ be an optimal small cover of $P$ with $q$ complete pieces.
Suppose for contradiction that $r\ge \left(1+\frac{C'}{\sqrt{k}}\right)q$.
Note that for $\varepsilon:=\frac{C'}{\sqrt{k}}\leq \frac{C'}{\sqrt{k_0}}\leq 1$, \Cref{lem:localsearchregions} guarantees  a set of $k'\leq C/\varepsilon^2=C\frac{k}{C'^2}=k$ pieces can be replaced by fewer pieces from $\mathcal Q$.
A contradiction to the fact that $\mathcal R$ is a $k$-local optimum. 
Hence, $r< \left(1+\frac{C'}{\sqrt{k}}\right)q$.
\end{proof}

\subsection{Polygons with Small Diameter}\label{sec:smalldiameter}

In this section, we describe a PTAS for small covers of polygons that intersect at most $D$ squares of the {integer} grid.
In order to obtain a PTAS, it is sufficient by \Cref{cor:localsearchcontinuous} to show that, for any fixed $k$, a $k$-local optimum can be found in polynomial time. The strategy is as follows: we start with some small cover of $P$, and then test every set of $k$ pieces in that cover to see if it can be replaced by a smaller set. If such a swap is possible, then we perform the swap and start over with the smaller cover. This repeats until no swap can be found. 

The polygon $P$ {with $n$ vertices} intersects at most $D$ {integer} grid squares, so we can find an initial cover with at most $Dn$ squares{--note that the intersection of a unit square with $P$ has at most $n$ connected components as each side of $P$ is incident to at most one connected component}.
So finding a $k$-local optimum requires at most a polynomial (in $n$ and $D$) number of swaps, because the number of pieces decreases in each step.
It remains to show how to find a swap, or determine that none exists.

Given a set of complete small pieces (\emph{not} necessarily a cover) in $P$, we define a corresponding ``combinatorial structure''. The combinatorial structure is determined by:

\begin{itemize}
\item The line arrangement containing edges of $P$ {and edges of the squares bounding the pieces}.
\item For each bounding square {$S$, the connected component of $P^\circ \cap S$ defining the piece (indicated with a number).}
\end{itemize} 

\begin{figure}[ht]
    \centering
    \includegraphics[page=4]{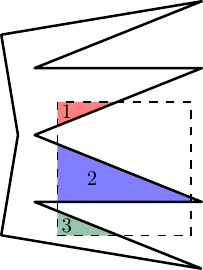}
    \hspace{3cm}
    \includegraphics[page=3]{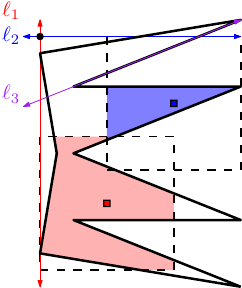}
    \caption{{The combinatorial structure of a set of complete small pieces. Left: The connected components of $P^\circ \cap S$ are numbered. Right: The centers of the bounding squares of the two pieces  are marked with squares. The intersection of lines $\ell_1$  and $\ell_2$ lies above $\ell_3$.}}
    \label{fig:combinatorial_structure}
\end{figure}

Let $\mathcal{L}$ be the set of lines containing edges of $P$ or edges of a bounding square.
For each $(\ell_1, \ell_2, \ell_3)\in \mathcal{L}^3$, the first piece of information says on which side of $\ell_3$ the intersection of $\ell_1$ and $\ell_2$ occurs, or if the intersection is contained in $\ell_3$. {(This is the orientation of the triangle formed by $(\ell_1, \ell_2, \ell_3)$ or the order type of the dual of the lines.)}
The slope of each line in $\mathcal{L}$ is fixed, so all the information about the arrangement of $\mathcal{L}$ can be described by equalities and (strict) inequalities that are linear in the coordinates of the bounding squares. For example, if $\ell_1$ is a vertical line containing the left edge of a square with center $(x_i, y_i)$, $\ell_2$ is a horizontal line containing the top edge of a square with center $(x_j, y_j)$, and $\ell_3=\{(x, y):\alpha x+\beta y=\gamma\}$ is a line containing an edge of $P$, then the combinatorial data contains one of:
\begin{itemize}
    \item $\alpha(x_i-\frac12)+\beta(y_j+\frac12)=\gamma$,
    \item $\alpha(x_i-\frac12)+\beta(y_j+\frac12)<\gamma$, or 
    \item $\alpha(x_i-\frac12)+\beta(y_j+\frac12)>\gamma$, 
\end{itemize}
specifying on which side of $\ell_3$ the intersection $(x_i-\frac12, y_j+\frac12)$ of $\ell_1$ and $\ell_2$ occurs {see \Cref{fig:combinatorial_structure}}.

A realization of a combinatorial structure is a set of pieces satisfying the constraints. Since the constraints are linear, the set of realizations of any fixed combinatorial structure is a convex set. 
Note that the combinatorial data does not ensure that a realization corresponds to a cover of the polygon $P$.
For instance, it is easy to make an example where all the squares are placed almost on top of each other, while big regions of $P$ remain uncovered.
Given just the combinatorial data of a (hypothetical) arrangement, we can use linear programming (LP) to check if it can actually be realized by some set of small pieces. What we really want to be able to check is if \emph{any} such realization corresponds to a cover of $P$. The following lemma says that we only need to check a single realization.

\begin{lemma}\label{lem:structureconvex}
For a given combinatorial structure, either 
each or no realization is a cover.
\end{lemma}

\begin{proof}
Let $A$ and $B$ {be} realizations of a combinatorial structure with $m$ squares. We show that $A$ is a cover if and only if $B$ is a cover. Suppose that $A$ is a cover. We continuously move the bounding squares from their positions in $A$ to their positions in $B$, along straight lines in the $2m$-dimensional coordinate space. If, during this process, the arrangement stops being a cover, then one of the following cases occurs: (a) part of a polygon edge is uncovered, (b) part of the interior is uncovered, or (c) a piece becomes disconnected. The  cases {are} illustrated in \Cref{fig:continuouscovering}. 
\begin{figure}[htb]
    \centering
    \begin{subfigure}{.33\textwidth}
    \centering
      \includegraphics[page=6]{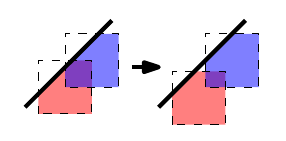}
    \subcaption{}
    \label{}
    \end{subfigure}\hfill
    \begin{subfigure}{.33\textwidth}
    \centering
    \includegraphics[page=7]{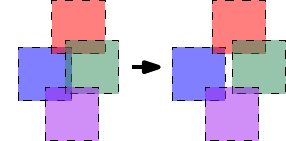}
    \subcaption{}
    \label{}
    \end{subfigure}\hfill
    \begin{subfigure}{.33\textwidth}
    \centering
    \includegraphics[page=8]{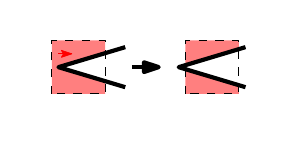}
    \subcaption{}
    \label{}
    \end{subfigure}\hfill
    \caption{If the bounding squares are moving continuously, then there are three ways that a covering might stop being valid. Any of these change the arrangement of lines in $\mathcal{L}$. (a) part of an edge stops being covered. (b) part of the interior stops being covered. (c) a piece ceases to be connected.}
    \label{fig:continuouscovering}
\end{figure}
We consider a snapshot right after the first point becomes uncovered. In case (a), the uncovered region is a triangle (or a rectangle). Note that the intersection of a right angle of this region has changed sides with respect to the polygon side independent of whether both or just one involved square is moving.
In case (b), the uncovered region is a rectangle and its appearance is due to two squares that become disjoint due to the movement of at least one square. In case (c), a side of a moving square now intersects a side of the polygon which was previously disjoint.
Consequently, in each case some of the combinatorial data is changing. Since the set of realizations of a combinatorial structure is convex, we conclude that this doesn't happen. Since the arrangement forms a cover at every point in the motion, we conclude that $B$ must be a cover.

By a symmetric argument, it is clear that if $B$ is a cover, then $A$ is a cover.
\end{proof}

Given a cover with complete pieces, we want to check if some set of $k$ of its pieces can be swapped for a set of $k-1$ pieces. This is accomplished by the next lemma.

\begin{lemma}\label{lem:singleswap}
Let $P$ be a polygon with $n$ vertices and consider a small cover $\mathcal{R}$ of $P$ with $m$ pieces.
For each fixed $k$, we can determine in {$\mathcal{O}_k((n+m)^{O(k)}) $ time} whether any set of $k$ of the pieces from the cover can be replaced by $k-1$ small pieces (not necessarily from the cover) to form a cover with fewer pieces.
\end{lemma}

\begin{proof}
For each of the $m\choose k$ $k$-subsets of $\mathcal{R}$, we remove it and test each of the ways of adding $k-1$ squares to the combinatorial arrangement.
We need to argue that for any fixed $k$, there are only polynomially many ways of adding $k-1$ squares to a combinatorial structure with $m$ squares. This corresponds to adding $2k$ vertical lines and $2k$ horizontal lines to the arrangement. A line arrangement with $\ell$ lines contains at most $O(\ell^2)$ intersection points.
{Thus} there are $O(\ell^2)$ combinatorially different ways to add a new line with a fixed slope.
A combinatorial structure with $m$ squares has $n+2m$ lines.
So the number of ways of adding the new $4k$ 
lines is $\mathcal{O}_k((n+m)^{4k})$. The intersection of a square with $P$ can have at most $n$ connected components, so in total there are $\mathcal{O}_k((n+m)^{4k+1})$ ways of adding $k$ squares to the arrangement. {For fixed $k$,} this is polynomial in $n$ and $m$. 

For each of these candidate combinatorial structures, we need to check {whether} it represents a valid covering. By \Cref{lem:structureconvex}, it {suffices} to find a single representative set of pieces (if one exists). 
Each element of the combinatorial data is a linear constraint involving the coordinates of the squares and either an equality or a strict inequality.
We construct an LP instance from these linear constraints, adding a new variable $t$ and replacing each strict inequality $\mathcal{C}>0$ with $\mathcal{C}\ge t$.
To test if a solution exists, we solve the LP instance  in which the objective function is to maximize $t$.
If there is no feasible solution with $t>0$, then we conclude that there is no valid swap with this combinatorial structure. If we find a feasible solution with $t>0$, then we check if this solution covers all of $P$.
We do this by constructing the pieces explicitly and computing their union.
If the union does not cover $P$, then we know that no arrangement with this combinatorial structure does.
The only free variables in the LP instance are $t$ and the $2(k-1)$ coordinates of the newly-added bounding squares. So in order to solve the LP instance, we can just check all of the $(n+m)^{\mathcal{O}(k)}$ basic solutions. This takes (strongly-)polynomial time for a fixed $k$.

We should check that this procedure {finds} a swap if one exists. Suppose that there is a set of $k$ pieces in the initial cover that can be swapped with $k-1$ pieces in order to obtain a new cover $\mathcal{R}'$.
The cover $\mathcal{R}'$ has a combinatorial structure, which is one of the candidate structures tested by our procedure.
Furthermore, $\mathcal{R}'$ is a realization of this structure with the positions of the $m-k$ old pieces being the same as in $\mathcal{R}$, so such a realization exists, and the linear programming step will find one (but the realization found is not necessarily the same one as $\mathcal{R}'$). Since $\mathcal{R}'$ is a cover, by \Cref{lem:structureconvex}, the realization found will be a cover.
\end{proof}

\begin{theorem}\label{thm:ptasmain}
There is an algorithm that, given a polygon $P$ and a parameter $\eps$, computes a $(1+\eps)$-approximation of the optimal small cover of $P$.
The algorithm runs in time $(Dn)^{\text{poly}(\frac{1}{\eps})}$, where $D$ is the number of {integer} grid squares intersecting $P$.
\end{theorem}

\begin{proof}
Fix $k$ such that $\frac{C'}{\sqrt{k}}\le \eps$, and $k\ge C'$, where $C'$ is the constant from \Cref{cor:localsearchcontinuous}.

We can construct a cover of $P$ with small pieces by overlaying $P$ with the {integer} grid and taking the connected components of each intersection.
{Recall} that the intersection of a unit square with $P$ has at most $n$ connected components.
Hence, we obtain an initial solution with at most $Dn$ pieces, where $D$ is the number of {integer} grid squares intersecting $P$. We then check if this cover is a $k$-local optimum using the procedure from \Cref{lem:singleswap}.

If this cover is not a $k$-local optimum, then we perform the swap and repeat the process until no more valid swaps remain. Every swap decreases the total number of pieces by $1$, so the algorithm terminates after at most $Dn$ steps, and when it does we have obtained a cover such that no set of $k$ pieces can be swapped for a strictly smaller set of pieces. By \Cref{cor:localsearchcontinuous}, this solution has at most $(1+\eps)$-times as many pieces as an optimal solution. It is clear by \Cref{lem:singleswap} that the entire process takes time at most $(Dn)^{\text{poly}(\frac{1}{\eps})}$.
\end{proof}

The algorithm in \Cref{thm:ptasmain} runs in polynomial time on real RAM, but the bit-complexity of the coefficients can grow by a constant factor after each swap, potentially leading to exponential coefficient growth.
In order to obtain an algorithm that runs in strongly-polynomial time, we can adjust the {coordinates} after each swap.
{Let $m$ be the number of squares of the cover.
We consider the LP instance $L$ corresponding to the combinatorial structure of the cover.
Here, $L$ has two variables for all $m$ squares, not just the new squares introduced in the latest swap.
The cover that we found is a feasible solution $x\in\mathbb R^{2m+1}$ to $L$.
There are two coordinates in $x$ for each square and one for the value $t$ in the strict inequalities.
We want to find a \emph{basic} feasible solution to $L$, i.e., a vertex of the polyhedron defined by the constraints, because the bit-complexity of the variables in such a solution is polynomial in $n$, $D$ and the bit-complexity of the coordinates of vertices of $P$.
The value of $t$ (the variable for strict inequalities) should still be positive in the new solution.}

This can be done in strongly-polynomial time.
We can write $L$ such that all constraints have form $\mathcal{C}\ge 0$.
We start by evaluating each constraint at the feasible solution $x$.
A constraint $\mathcal{C}\ge 0$ is \emph{tight} if $\mathcal{C}(x)=0$.
{Let $\mathcal D$ be the set of tight constraints and let $\mathcal H=\{y\in\mathbb R^{2m+1}\mid \forall \mathcal C\in\mathcal D: \mathcal C(y)=0\}$.
In other words, $\mathcal H$ is the affine subspace that causes the tight constraints to be tight.}
The solution $x$ is a basic feasible solution if $\mathcal{H}$ has dimension $0$.
Otherwise, we can find a direction $v\in \text{Vec}(\mathcal{H})$ such that $t$ does not decrease when moving in direction $v$.
We move from $x$ in direction $v$ until an additional constraint becomes tight. 

{We obtain a new feasible solution $x'$ to $L$, where $t$ at $x'$ is at least as large as $t$ at $x$.
Furthermore, the dimension of $\mathcal{H}'$, the subspace of tight constraints corresponding to $x'$, has decreased by at least $1$ compared to $\mathcal H$.}
So we can repeat this process to find a basic feasible solution with the required properties.
The number of arithmetic operations depends only on the number of variables and the number of constraints, so this process runs in strongly polynomial time.
The bit-complexity of a basic feasible solution is polynomial in $n$, $D$ and the bit-complexity of the coordinates of vertices of $P$.
In particular, it does not increase with the number of swaps performed by the algorithm in \Cref{thm:ptasmain}.

\subsection{Polygons with Large Diameter}\label{sec:ptas:largepolygons}

In this section, we assume the existence of an algorithm $\mathcal{A}$ providing a $(1+\varepsilon)$-approximation for a small cover of a polygon $P$ with diameter at most $O(n/\eps)$. We show that from such an algorithm, one can design an algorithm $\mathcal{A}'$ that for an arbitrary polygon $P$, for which the size of a minimum cover of $P$ is $\OPT$, can find a value $\lambda$ such that $\OPT\leq \lambda\leq (1+\eps_0)\OPT$, where $\eps_0\leq \alpha\eps$ for some universal constant $\alpha$. The algorithm also provides an implicit description of a cover with this approximation guarantee. Recall that we work in the real RAM model. In particular, we can perform basic arithmetic operations like addition, multiplication, and division on the coordinates of the corners of the polygon in constant time.
\begin{restatable}{theorem}{largeDiameter}\label{thm:sparse-reduction}
Assume there exists an algorithm $\mathcal{A}$ which for any polygon $P$ with $n$ vertices and diameter $O(n/\eps)$ finds a $(1+\eps)$-approximate minimum cover of $P$ in time $f(n,\eps)$. 
Then there exists an algorithm $\mathcal{A}'$ that for any $n$-vertex polygon $P$ returns a value $\lambda_P$ with  $\OPT\leq \lambda_P\leq (1+\alpha\eps)\OPT$, where $\OPT$ is the size of a minimal cover of $P$ with small pieces and $\alpha$ is a sufficiently large universal constant. The runtime of $\mathcal{A}'$ is
$
f(\alpha n,\eps)\cdot \poly(n/\eps)$.
The algorithm also produces an implicit description of a $(1+\alpha\eps)$-approximate cover.
For any query point $q\in P$, the algorithm can in time $O(\poly(n/\eps)+f(4,\eps))$ return all pieces in the cover containing $q$.
 \end{restatable}

 \begin{corollary}
    There is an algorithm that, given a polygon $P$ and a parameter $0<\eps<1$, computes a $(1+\eps)$-approximation to the number of pieces in an optimal small cover of $P$. The algorithm runs in time $(n/\eps)^{\poly(1/\eps)}$.

    The algorithm also produces an implicit description of a $(1+\eps)$-approximate cover.
For any query point $q\in P$, the algorithm can in time $O(\poly(n/\eps)+ (1/\eps)^{\poly(1/\eps)}))$ return all pieces in the cover containing $q$.
\end{corollary}
\begin{proof}
Use~\cref{thm:sparse-reduction} with $\mathcal{A}$ being the algorithm of~\cref{thm:ptasmain} with $D=O((n/\eps)^2)$. 
\end{proof}

We now sketch the idea of the reduction; the details are given in \Cref{sec:bigpoly}.
We first argue that there is a set of $k=O(n)$  disjoint `large' \emph{corridors} $\mathcal{C}=\{C_1,\dots,C_k\}$  with each $C_i\subset P$ such that we can solve the covering problem on $C_1,\dots,C_k,$ and $P\setminus \bigcup \mathcal C$ separately. Each corridor $C_i$ is constructed to have a pair of sides that are either both vertical or both horizontal and are of distance $\geq 1/\eps$ apart. The remaining two sides are both contained in an edge of $P$. Finally, each $C_i$ is essentially picked to be of maximum volume satisfying these properties. See~\cref{fig:rhombus}.

We will show that the union of $(1+O(\eps))$-approximate covers of $C_1,\dots,C_k$ and $P\setminus \bigcup\mathcal C$ yields a $(1+O(\eps))$-approximate cover of $P$. For  $P\setminus \bigcup\mathcal C$,  the diameter of each connected component is at most $n/\varepsilon$, so we can use the algorithm $\mathcal{A}$ for bounded diameter polygons.
The corridors are more tricky as the size of their optimal solutions cannot be bounded in terms of $n$, so for those, we maintain an implicit solution. Denote by $a_i$ and $p_i$ the area and perimeter of $C_i$.
If $a_i\geq p_i/\eps$, one can show that the intersection $C_i$ with the {integer} grid gives an $1+O(\eps)$ approximate cover, and this solution is easy to describe implicitly. The tricky part is handling a corridor $C_i$ with $a_i\leq p_i/\eps$ where $C_i$ is `long' and `skinny'. Suppose that $C_i$ has a pair of vertical edges (the horizontal case is similar). In this case, we partition $C_i$ further into corridor sections $(S_i^j)_{j=1}^t$, all having a pair of vertical edges and width $1/\varepsilon$. We then again argue that it suffices to solve the problem on each $S_i^j$ with an approximation ratio of $1+O(\varepsilon)$. Unfortunately, we cannot afford  to find a solution for each $S_i^j$ separately, but we can argue about a certain monotonicity of sizes of optimal covers of the $S_i^j$. Namely, after appropriate translations, it is either the case that $S_i^j\subset S_i^{j+1}$ for all $j$ or it is the case that $S_i^j\supset S_i^{j+1}$ for all $j$. In the former case, we show that the size of the optimal cover of $S_i^j$ is at most the size of the optimal cover of $S_i^{j+1}$ and a similar claim holds in the latter case. 
Based on these properties, and defining $\ell=\eps^4t$, we show that it suffices to approximate the minimum covers of $S_i^{\ell},S_i^{2\ell},\dots S_i^{\lfloor t/\ell\rfloor \ell}$.
For any $S_i^j$ with $k\ell \leq j<(k+1)\ell$, for some $k$, we can then use the size of the approximate cover of $S_i^{(k+1)\ell}$ as an upper bound on the size of the optimal cover of $S_i^j$. This only requires us to run $\mathcal{A}$ on the $t/\ell=O(1/\eps^4)$ instances of the form $S_i^{(k+1)\ell}$.
While for some values of $k$, this approach might not give $1+O(\eps)$ approximation to $S_i^j$, the monotonicity ensures that for most values of $k$, the approximation is good, and suffices to give a final approximation guarantee of $1+O(\eps)$.

\subsection{Equivalence of Small Covers and Partitions}\label{sec:equivalence}

The proof of~\Cref{lem:nondegenerate} lets us show the equivalence of small covering and partitioning -- a result which may be interesting in its own right.

\begin{theorem}\label{thm:=}
For any polygon $P$, possibly with holes, optimal small covers and small partitions have the same number of pieces.
\end{theorem}

\begin{proof}

Consider an optimal small cover $\mathcal Q$ consisting of complete small pieces; see \Cref{fig:partition}.

\begin{figure}[ht]
\centering
\includegraphics[page=13,width=\textwidth]{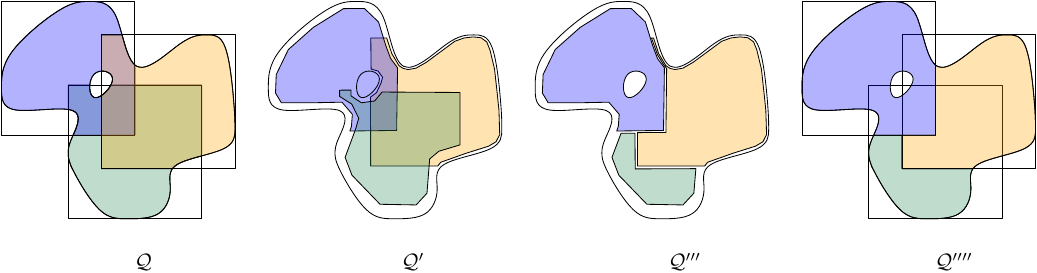}
\caption{Steps in the proof of \Cref{thm:=}.}
\label{fig:partition}
\end{figure}

Consider the overlay of $\mathcal Q$ and choose a point in the interior of each cell which is also in~$P$.
Construct a set of proper non-piercing regions $\mathcal Q'$ as in \Cref{lem:nondegenerate}.
As described in the proof, the pieces in $\mathcal Q'$ are subsets of those in $\mathcal Q$, except that some holes of $P$ have been filled.
Importantly, the boundaries of the pieces in $\mathcal Q'$ are just slight perturbations of those in $\mathcal Q$.

The following process of \emph{lens bypassing} is described in \cite[Section~4]{NonpiercingPTAS}.
Since $\mathcal Q'$ are non-piercing, any intersection $Q'_i\cap Q'_j$ consists of a set of \emph{lenses}, which are regions bounded by one interval of the boundaries of each set $Q'_i$ and $Q'_j$.
In the process of lens bypassing, we consider an inclusion-wise minimal lens $\ell$, say in the intersection $Q'_i\cap Q'_j$.
Let the boundary of $\ell$ be $\gamma_i\cup\gamma_j$, where $\gamma_i$ and $\gamma_j$ are on the boundaries of $Q'_i$ and $Q'_j$, respectively.
We change the boundary of $Q'_i$ by replacing $\gamma_i$ by a curve following $\gamma_j$ (but slightly outside $\ell$ to avoid non-proper intersections).
This effectively removes $\ell$ from $Q'_i$.
It is shown~\cite[Lemma 4.6]{NonpiercingPTAS} that the operation results in a set of regions that are also non-piercing.
Repeating the process results in a set $\mathcal Q''$ of pairwise disjoint regions that are subsets of the original regions $\mathcal Q'$ and still contain the points $X$.

Consider any hole $H$ of $P$, i.e., $H$ is a bounded connected component in the complement of $P$.
Recall that in the pieces $\mathcal Q'$, the hole $H$ might have been filled in some of them, so $H$ might also be filled in some of the pieces $\mathcal Q''$.
In our modifications from $Q$ to $Q''$, we have only ever made modifications along curves that are disjoint from $H$.
Hence, for each piece $Q''_i\in\mathcal Q''$, we either have $H\subset Q''_i$ ($H$ is filled) or $H\cap Q''_i=\emptyset$ ($H$ is still a hole of $Q''_i$ or outside the outer boundary).
We can remove $H$ from all pieces where it is filled, thus reintroducing $H$ as a hole, and this will not break connectivity.
We do so for all holes of $P$.
We end up with regions $\mathcal Q'''$ that are pairwise disjoint and all contained in $P$.
Furthermore, the regions $\mathcal Q'''$ still cover $P$ in the sense that each region in the overlay of the original cover $\mathcal Q$ is contained in a region in the overlay of $\mathcal Q'''$, except that the boundary might be slightly perturbed.
We can thus create pieces $\mathcal Q''''$ analogous to $\mathcal Q'''$ using the original boundaries, and these pieces will cover all of $P$ and still be pairwise interior-disjoint.
We therefore obtain a small partition of $P$.
\end{proof}

\section{Hardness for \SmallCover\ in 3 dimensions}\label{sec:hardness3D}

In this section, we give a gap-preserving reduction from \SetCover\  to $3$-dimensional small cover. The best-known inapproximability result for \SetCover\  is the following:

\begin{theorem}(Dinur and Steurer \cite{SetCoverInapprox})\label{thm:setcoverinnaprox}
For any $\alpha>0$, it is NP-hard to approximate an instance of \SetCover\ of size $n$ to within a factor of $(1-\alpha)\ln(n)$.
\end{theorem}

Our results show that \SmallCover\ in 3 dimensions is no easier to approximate than \SetCover. 

\begin{theorem}
Let $I$ be an instance of \SetCover\ where the universe is $U=\{1, \dots, n\}$ and the sets are $\mathcal{S}=\{S_1, \dots, S_{m}\}$. Then we can construct in polynomial time a simple polyhedron $P\subseteq \mathbb{R}^3$ with $\mathcal{O}\left(n+\sum_i |S_{i}|\right)$ vertices such that $P$ has a unit cube cover of size $k+5$ if and only if $U$ can be covered by $k$ of the sets in $\mathcal{S}$. The vertices of $P$ have bit-complexity at most $\mathcal{O}(\log(nm))$ and $P$ is contained in a box of size $1\times 2\times 3$.
\end{theorem}

\begin{proof}
For simplicity, we will scale the coordinates by a factor of $8$. So the polyhedron $P$ will have a cover with pieces that fit inside $8\times 8\times 8$-cubes. We will construct $P$ from a union of flat plane pieces, but it can easily be thickened to obtain a non-degenerate polyhedron. 

To construct $P$, we start by making layers of \emph{hook} pieces. There is one type of hook piece for each of the $m$ sets in $\mathcal{S}$. The hook piece for $S_j$ is shown in \Cref{fig:hookpiece}. For each $i\in S_j$, we place a copy of the corresponding hook piece at $z$-coordinate $\frac{mi+j}{mn}$. 

\begin{figure}[htb]
	\centering   
    \begin{subfigure}[t]{.47\textwidth}
        \centering  
        \includegraphics[page=1]{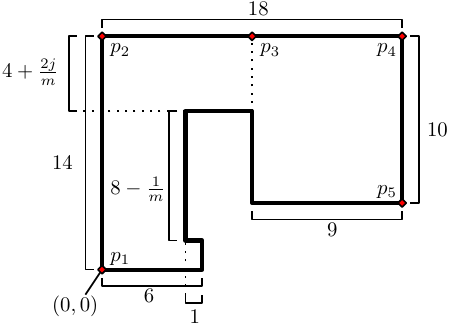}
	\subcaption{}
	\label{fig:hookpiece}
    \end{subfigure}\hfill
    \begin{subfigure}[t]{.47\textwidth}
        \centering   
	\includegraphics[page=2]{figures/SimplePolyhedraHardness-1.pdf}
	\subcaption{}
	\label{fig:wallpieces}
    \end{subfigure}
	\caption{
    (a) Specification of a hook piece. The  $(x, y)$-positions of the points $p_1$ through $p_5$ are also shown.
    (b) The segments marked show the wall pieces that connect the hook pieces together. Segments $W_1$ through $W_5$ extend to cover the height of the entire figure, while segment $W_6$ is extended over $n$ disjoint intervals. The outline of a hook piece is shown in green.
    }
\end{figure}

The hook pieces are connected by a wall, which is formed by the segments marked $W_1$ through $W_5$ in \Cref{fig:wallpieces} extruded to cover the $z$-coordinates $[\frac{m+1}{mn}, \frac{(m+1)n}{mn}]$. Additionally, the segment marked $W_6$ is extruded to cover the $z$ coordinates $[\frac{mi+1}{mn}, \frac{(m+1)i}{mn}]$ for each $i\in U$. The entire construction is illustrated in \Cref{fig:3dpolyhedron}.

\begin{figure}[htb]
	\centering   
	\includegraphics[page=3]{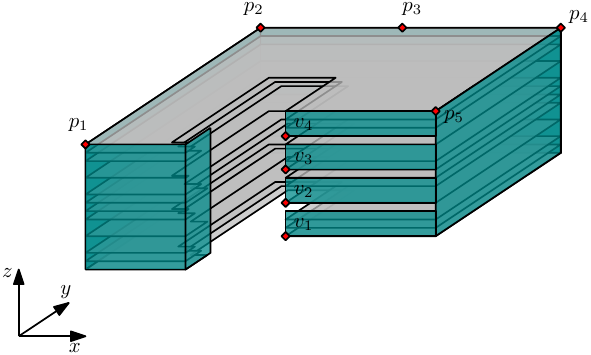}
	\caption{A rendering of the polyhedron $P$. The specific instance shown corresponds to $n=4$, $m=4$, $S_1=\{1, 2, 3\}$, $S_2=\{1, 3, 4\}$, $S_3=\{1, 2, 4\}$, and $S_4=\{2, 3, 4\}$. The hook pieces are shown in gray and the wall pieces are shown in green. The vertical scale is exaggerated for clarity.}
	\label{fig:3dpolyhedron}
\end{figure}

Suppose that $U$ can be covered by $k$ of the sets in $\mathcal{S}$. We want to show that $P$ can be covered with $k+5$ pieces that fit inside an $8\times 8\times 8$-unit cube. First note that the $5$ cubes shown in \Cref{fig:5cubes} are sufficient to cover most of $P$. 

\begin{figure}[htb]
	\centering   
    
	\includegraphics[page=4]{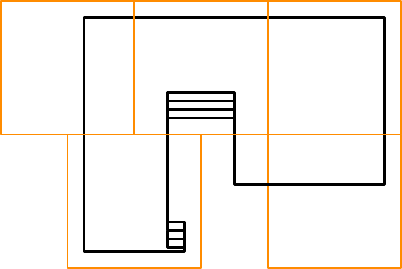}
	\caption{These $5$ cubes cover most of $P$.}
	\label{fig:5cubes}
\end{figure}

The intersection of $P$ with each of these cubes is connected. There are $n$ connected chunks remaining uncovered, corresponding to the $n$ elements of $U$. Let $Q_i$ be the connected part of the uncovered region containing $v_i$. For each $S_j$, there is a connected piece containing $\bigcup_{i\in S_j}Q_i$ that fits inside the cube $[4, 12]\times [2-\frac{4j-1}{2m}, 10-\frac{4j-1}{2m}]\times [0, 8]$, as illustrated in \Cref{fig:hookpath,fig:hookcomponent}. So if $U$ can be covered by $k$ of the $S_j$, then $P$ can be covered by $k+5$ pieces.

\begin{figure}[htb]
	\centering   
    \begin{subfigure}[t]{.47\textwidth}
       \centering
       \includegraphics[page=5]{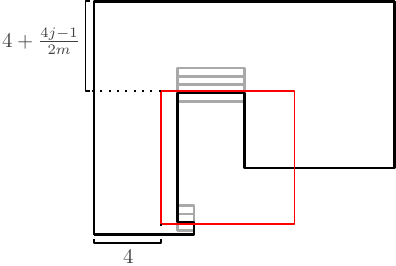}
	\subcaption{}
	\label{fig:hookpath}
    \end{subfigure}\hfill
	\begin{subfigure}[t]{.47\textwidth}
    \centering   
	\includegraphics[page=6]{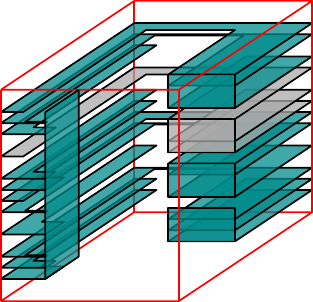}
	\subcaption{}
    \label{fig:hookcomponent}
    \end{subfigure}
    \caption{
        (a) The intersection of the cube $[4, 12]\times [2-\frac{4j-1}{2m}, 10-\frac{4j-1}{2m}]\times [0, 8]$ with a hook piece corresponding to $S_j$ connects some of the $Q_i$.
        (b) A 3-dimensional view of the intersection of $P$ with the cube $[4, 12]\times [2-\frac{4j-1}{2m}, 10-\frac{4j-1}{2m}]\times [0, 8]$. Shown is the example from \Cref{fig:3dpolyhedron} with $j=3$. The green part is connected and contains $Q_1$, $Q_2$ and $Q_4$. The gray part contains $Q_3$, but cannot be connected to the green part without leaving the cube.
    }
\end{figure}

Now suppose that $P$ can be covered by $k+5$ pieces. We first designate five points $p_1,\ldots,p_5$ (as shown in \Cref{fig:hookpiece,fig:3dpolyhedron}). No $8\times 8\times 8$ cube can contain more than $1$ of these points, so there are at least five pieces in the covering, one for each of these points. These pieces also cannot contain any of the points $v_1,\ldots,v_n$, so the remaining pieces must cover these points.

Suppose that some piece $Q$ in the cover contains $\{v_i : i\in S\}$ for some set $S\subseteq U$. We want to show that $S\subseteq S_j$ for some $j$. By assumption, each $i\in U$ is in some $S_j$, so WLOG suppose $|S|\ge 2$. Let $i, \ell\in S$ with $i\ne \ell$. So $Q$ contains a path $\gamma$ from $v_i$ to $v_\ell$. No part of the walls $W_3$ through $W_5$ is close enough to $v_i$ to be contained in $Q$, so $\gamma$ must connect through $W_1$ or $W_2$.

In order to get from $v_i$ to $W_1$ or $W_2$, $\gamma$ must go through some hook piece, say this hook piece corresponds to $S_j$. Such a connection contains $y$ coordinates covering the interval $[2-\frac{2j}{m}, 10-\frac{2j-1}{m}]$ (see \Cref{fig:hookpath,fig:hookcomponent}). Since $Q$ is contained in an $8\times 8\times 8$-unit cube, it cannot contain points with $y$-coordinates covering $[2-\frac{2j'}{m}, 10-\frac{2j'-1}{m}]$ for any $j'\ne j$. So any $v_\ell$ contained in $Q$ is connected to $v_i$ through a hook piece corresponding to $S_j$. This is only possible if $\ell\in S_j$. So $S\subseteq S_j$.

Thus, $U$ can be covered by $k$ of the $S_j$ if and only if $P$ can be covered by $k+5$ connected pieces each fitting inside an $8\times 8\times 8$-unit cube.

If we thicken $P$ slightly, we obtain a simple (topologically ball-shaped) polyhedron. To see this, we observe that $P$ is contractable. Indeed, each hook piece (strongly) deformation retracts onto walls $W_1$ through $W_6$. So all of the hook pieces deformation retract onto the wall pieces. Since the union of the wall pieces is contractable, $P$ is contractable.
\end{proof}

We note that it is possible to trim the cover of size $k+5$ described in the proof in order to obtain a partition, so we get hardness for both \SmallCover\ and \SmallPartition.
In fact, it is easy to see that the problems are equivalent in $3$ and higher dimensions:
If $Q_1,\ldots,Q_q$ form a cover for a polyhedron $P$, we can define $Q'_i=Q_i\setminus \bigcup_{j=1}^{i-1} Q_j$ for $i=1,\ldots,q$.
Then the pieces $Q'_i$ are pairwise disjoint and also cover $P$, but they may be disconnected.
However, if $Q'_i$ is disconnected, we can connect the parts by thin tunnels, staying in the original piece $Q_i$ so we don't violate the size constraint.
These tunnels are removed from the other pieces $Q'_j$, and choosing them sufficiently thin, they will not disconnect any piece.

\begin{corollary}
In 3 dimensions, for any $\alpha>0$, it is NP-hard to approximate \SmallCover\ and \SmallPartition\ better than $(1-\alpha)\ln(n)$, where $n$ is the number of bits needed to describe the input polyhedron.
This holds even when the polyhedron is simple and contained in a box of size $1\times 2\times 3$.
\end{corollary}

\bibliography{bib}

\begin{thebibliography}{10}

\bibitem{aamand2023tiling}
Anders Aamand, Mikkel Abrahamsen, Thomas~D. Ahle, and Peter M.~R. Rasmussen.
\newblock Tiling with squares and packing dominos in polynomial time.
\newblock {\em {ACM} Trans. Algorithms}, 19(3):30:1--30:28, 2023.
\newblock \href {https://doi.org/10.1145/3597932} {\path{doi:10.1145/3597932}}.

\bibitem{DBLP:conf/focs/Abrahamsen21}
Mikkel Abrahamsen.
\newblock Covering polygons is even harder.
\newblock In {\em Symposium on Foundations of Computer Science (FOCS)}, pages
  375--386. {IEEE}, 2021.
\newblock \href {https://doi.org/10.1109/FOCS52979.2021.00045}
  {\path{doi:10.1109/FOCS52979.2021.00045}}.

\bibitem{DBLP:journals/jacm/AbrahamsenAM22}
Mikkel Abrahamsen, Anna Adamaszek, and Tillmann Miltzow.
\newblock The art gallery problem is $\exists\mathbb {R}$-complete.
\newblock {\em J. {ACM}}, 69(1):4:1--4:70, 2022.
\newblock \href {https://doi.org/10.1145/3486220} {\path{doi:10.1145/3486220}}.

\bibitem{DBLP:conf/stoc/AbrahamsenBNZ24}
Mikkel Abrahamsen, Joakim Blikstad, Andr{\'{e}} Nusser, and Hanwen Zhang.
\newblock Minimum star partitions of simple polygons in polynomial time.
\newblock In {\em Symposium on Theory of Computing, (STOC)}, pages 904--910.
  {ACM}, 2024.
\newblock \href {https://doi.org/10.1145/3618260.3649756}
  {\path{doi:10.1145/3618260.3649756}}.

\bibitem{DBLP:journals/corr/abs-2211-01359}
Mikkel Abrahamsen and Nichlas~Langhoff Rasmussen.
\newblock Partitioning a polygon into small pieces.
\newblock In {\em Symposium on Discrete Algorithms (SODA)}, pages 3562--3589.
  2025.
\newblock \href {https://doi.org/10.1137/1.9781611978322.118}
  {\path{doi:10.1137/1.9781611978322.118}}.

\bibitem{abrahamsen2026hardness}
Mikkel Abrahamsen and Jack Stade.
\newblock Hardness of packing, covering and partitioning simple polygons with
  unit squares.
\newblock {\em SIAM Journal on Computing}, pages FOCS24--29, 2026.
\newblock \href {https://doi.org/10.1137/24M171824X}
  {\path{doi:10.1137/24M171824X}}.

\bibitem{DBLP:conf/esa/ArkinD0GMPT20}
Esther~M. Arkin, Rathish Das, Jie Gao, Mayank Goswami, Joseph S.~B. Mitchell,
  Valentin Polishchuk, and Csaba~D. T{\'{o}}th.
\newblock Cutting polygons into small pieces with chords: Laser-based
  localization.
\newblock In {\em European Symposium on Algorithms ({ESA})}, pages 7:1--7:23,
  2020.
\newblock \href {https://doi.org/10.4230/LIPIcs.ESA.2020.7}
  {\path{doi:10.4230/LIPIcs.ESA.2020.7}}.

\bibitem{SeparatorLocalSearch}
Rom Aschner, Matthew~J. Katz, Gila Morgenstern, and Yelena Yuditsky.
\newblock Approximation schemes for covering and packing.
\newblock In {\em WALCOM: Algorithms and Computation}, pages 89--100, 2013.
\newblock \href {https://doi.org/10.1007/978-3-642-36065-7\_10}
  {\path{doi:10.1007/978-3-642-36065-7\_10}}.

\bibitem{DBLP:conf/compgeom/BeckerFKLMS13}
Aaron~T. Becker, S{\'{a}}ndor~P. Fekete, Alexander Kr{\"{o}}ller, SeoungKyou
  Lee, James McLurkin, and Christiane Schmidt.
\newblock Triangulating unknown environments using robot swarms.
\newblock In {\em Symposium on Computational Geometry (SoCG)}, pages 345--346,
  2013.
\newblock \href {https://doi.org/10.1145/2462356.2462360}
  {\path{doi:10.1145/2462356.2462360}}.

\bibitem{bezdek1995solution}
A.~Bezdek and K.~Bezdek.
\newblock A solution of conway's fried potato problem.
\newblock {\em Bulletin of the London Mathematical Society}, 27(5):492--496,
  1995.
\newblock \href {https://doi.org/10.1112/blms/27.5.492}
  {\path{doi:10.1112/blms/27.5.492}}.

\bibitem{bezdek1996conway}
A.~Bezdek and K.~Bezdek.
\newblock Conway's fried potato problem revisited.
\newblock {\em Arch. Math.}, 66(6):522--528, 1996.
\newblock \href {https://doi.org/10.1007/BF01268872}
  {\path{doi:10.1007/BF01268872}}.

\bibitem{borsuk1933drei}
Karol Borsuk.
\newblock Drei {S}{\"a}tze {\"u}ber die $n$-dimensionale euklidische
  {S}ph{\"a}re.
\newblock {\em Fundamenta Mathematicae}, 20(1):177--190, 1933.
\newblock \href {https://doi.org/10.4064/fm-20-1-177-190}
  {\path{doi:10.4064/fm-20-1-177-190}}.

\bibitem{buchin2021decomposing}
Maike Buchin and Leonie Selbach.
\newblock Decomposing polygons into fat components.
\newblock In {\em Canadian Conference on Computational Geometry ({CCCG})},
  pages 175--184, 2021.

\bibitem{canete2022conway}
Antonio Ca\~{n}ete, Isabel Fern{\'a}ndez, and Alberto M{\'a}rquez.
\newblock Conway's fried potato problem: a (quadratic) algorithm leading to an
  optimal division for convex polygons.
\newblock In {\em Discrete Mathematics Days (DMD)}, pages 71--76, 2022.

\bibitem{chazelle1985approximation}
Bernard Chazelle.
\newblock Approximation and decomposition of shapes.
\newblock In {\em Algorithmic and Geometric Aspects of Robotics}, volume~1 of
  {\em Advances in Robotics}, pages 145--185. 1987.

\bibitem{chazelle1985optimal}
Bernard Chazelle and David~P. Dobkin.
\newblock Optimal convex decompositions.
\newblock In {\em Computational Geometry}, volume~2 of {\em Machine
  Intelligence and Pattern Recognition}, pages 63--133. 1985.
\newblock \href {https://doi.org/10.1016/B978-0-444-87806-9.50009-8}
  {\path{doi:10.1016/B978-0-444-87806-9.50009-8}}.

\bibitem{chazelle1994decomposition}
Bernard Chazelle and Leonidas Palios.
\newblock Decomposition algorithms in geometry.
\newblock In {\em Algebraic Geometry and its applications}, pages 419--447.
  1994.
\newblock \href {https://doi.org/10.1007/978-1-4612-2628-4_27}
  {\path{doi:10.1007/978-1-4612-2628-4_27}}.

\bibitem{chen2022skeleton}
Xun Chen and Mingyong Chen.
\newblock Skeleton partition models for {3D} printing.
\newblock In {\em Second International Conference on Medical Imaging and
  Additive Manufacturing (ICMIAM 2022)}, 2022.
\newblock \href {https://doi.org/10.1117/12.2636369}
  {\path{doi:10.1117/12.2636369}}.

\bibitem{croft2012unsolved}
Hallard~T. Croft, Kenneth~J. Falconer, and Richard~K. Guy.
\newblock {\em Unsolved Problems in Geometry}.
\newblock Problem Books in Mathematics: Unsolved Problems in Intuitive
  Mathematics. 1991.
\newblock \href {https://doi.org/10.1007/978-1-4612-0963-8}
  {\path{doi:10.1007/978-1-4612-0963-8}}.

\bibitem{damian2004computing}
Mirela Damian and Sriram~V. Pemmaraju.
\newblock Computing optimal diameter-bounded polygon partitions.
\newblock {\em Algorithmica}, 40(1):1--14, 2004.
\newblock \href {https://doi.org/10.1007/s00453-004-1092-3}
  {\path{doi:10.1007/s00453-004-1092-3}}.

\bibitem{DBLP:journals/algorithmica/BergKMT23}
Mark de~Berg, S{\'{a}}ndor Kisfaludi{-}Bak, Morteza Monemizadeh, and Leonidas
  Theocharous.
\newblock Clique-based separators for geometric intersection graphs.
\newblock {\em Algorithmica}, 85(6):1652--1678, 2023.
\newblock \href {https://doi.org/10.1007/S00453-022-01041-8}
  {\path{doi:10.1007/S00453-022-01041-8}}.

\bibitem{SetCoverInapprox}
Irit Dinur and David Steurer.
\newblock Analytical approach to parallel repetition.
\newblock In {\em Symposium on Theory of Computing (STOC)}, page 624–633.
  ACM, 2014.
\newblock \href {https://doi.org/10.1145/2591796.2591884}
  {\path{doi:10.1145/2591796.2591884}}.

\bibitem{fekete2011exploring}
S{\'{a}}ndor~P. Fekete, Tom Kamphans, Alexander Kr{\"{o}}ller, Joseph S.~B.
  Mitchell, and Christiane Schmidt.
\newblock Exploring and triangulating a region by a swarm of robots.
\newblock In {\em Approximation, Randomization, and Combinatorial Optimization.
  Algorithms and Techniques ({APPROX-RANDOM})}, pages 206--217. 2011.
\newblock \href {https://doi.org/10.1007/978-3-642-22935-0\_18}
  {\path{doi:10.1007/978-3-642-22935-0\_18}}.

\bibitem{DensitySeparators}
Sariel Har-Peled and Kent Quanrud.
\newblock Approximation algorithms for polynomial-expansion and low-density
  graphs.
\newblock {\em SIAM Journal on Computing}, 46(6):1712--1744, 2017.
\newblock \href {https://doi.org/10.1137/16M1079336}
  {\path{doi:10.1137/16M1079336}}.

\bibitem{DBLP:journals/combinatorics/JenrichB14}
Thomas Jenrich and Andries~E. Brouwer.
\newblock A 64-dimensional counterexample to {B}orsuk's conjecture.
\newblock {\em Electron. J. Comb.}, 21(4):4, 2014.
\newblock \href {https://doi.org/10.37236/4069} {\path{doi:10.37236/4069}}.

\bibitem{jiang2017models}
Xiaotong Jiang, Xiaosheng Cheng, Qingjin Peng, Luming Liang, Ning Dai,
  Mingqiang Wei, and Cheng Cheng.
\newblock Models partition for {3D} printing objects using skeleton.
\newblock {\em Rapid Prototyping Journal}, 23:54--64, 2017.
\newblock \href {https://doi.org/10.1108/RPJ-07-2015-0091}
  {\path{doi:10.1108/RPJ-07-2015-0091}}.

\bibitem{kahn1993counterexample}
Jeff Kahn and Gil Kalai.
\newblock A counterexample to {B}orsuk’s conjecture.
\newblock {\em Bulletin of the American Mathematical Society}, 29(1):60--62,
  1993.
\newblock \href {https://doi.org/10.1090/S0273-0979-1993-00398-7}
  {\path{doi:10.1090/S0273-0979-1993-00398-7}}.

\bibitem{keil1999polygon}
J.~Mark Keil.
\newblock Polygon decomposition.
\newblock In J{\"o}rg-R{\"u}diger Sack and Jorge Urrutia, editors, {\em
  Handbook of computational geometry}, chapter~11, pages 491--518. 1999.
\newblock \href {https://doi.org/10.1016/B978-044482537-7/50012-7}
  {\path{doi:10.1016/B978-044482537-7/50012-7}}.

\bibitem{keil1985minimum}
J.~Mark Keil and J{\"o}rg-R. Sack.
\newblock Minimum decompositions of polygonal objects.
\newblock In {\em Computational Geometry}, volume~2 of {\em Machine
  Intelligence and Pattern Recognition}, pages 197--216. 1985.
\newblock \href {https://doi.org/10.1016/B978-0-444-87806-9.50012-8}
  {\path{doi:10.1016/B978-0-444-87806-9.50012-8}}.

\bibitem{KO2025104759}
Minseok Ko, Yeongjun Yoon, Jaeyeon Kim, Samyeon Kim, and Soonjo Kwon.
\newblock {D-ECO}mposer: Sustainable part decomposition for additive
  manufacturing using machine learning based life cycle assessment.
\newblock {\em Additive Manufacturing}, 103:104759, 2025.
\newblock \href {https://doi.org/10.1016/j.addma.2025.104759}
  {\path{doi:10.1016/j.addma.2025.104759}}.

\bibitem{DBLP:journals/tog/LuoBRM12}
Linjie Luo, Ilya Baran, Szymon Rusinkiewicz, and Wojciech Matusik.
\newblock Chopper: {P}artitioning models into {3D}-printable parts.
\newblock {\em {ACM} Trans. Graph.}, 31(6):129:1--129:9, 2012.
\newblock \href {https://doi.org/10.1145/2366145.2366148}
  {\path{doi:10.1145/2366145.2366148}}.

\bibitem{o1987art}
Joseph O'Rourke.
\newblock {\em Art Gallery Theorems and Algorithms}.
\newblock Oxford University Press, 1987.

\bibitem{o2004polygons}
Joseph O’Rourke, Subash Suri, and Csaba~D. T\'{o}th.
\newblock Polygons.
\newblock In J.~E. Goodman and J.~O'Rourke, editors, {\em Handbook of discrete
  and computational geometry}, chapter~30, pages 787--810. Third edition
  edition, 2018.
\newblock \href {https://doi.org/10.1201/9781315119601}
  {\path{doi:10.1201/9781315119601}}.

\bibitem{NonpiercingPTAS}
Aniket~Basu Roy, Sathish Govindarajan, Rajiv Raman, and Saurabh Ray.
\newblock Packing and covering with non-piercing regions.
\newblock {\em Discret. Comput. Geom.}, 60(2):471--492, 2018.
\newblock \href {https://doi.org/10.1007/S00454-018-9983-2}
  {\path{doi:10.1007/S00454-018-9983-2}}.

\bibitem{shermer1992recent}
Thomas~C. Shermer.
\newblock Recent results in art galleries.
\newblock {\em Proceedings of the IEEE}, 80(9):1384--1399, 1992.
\newblock \href {https://doi.org/10.1109/5.163407}
  {\path{doi:10.1109/5.163407}}.

\bibitem{DBLP:journals/cgf/VanekGBMCSM14}
Juraj Vanek, Jorge A.~Garcia Galicia, Bedrich Benes, Radom{\'{\i}}r Mech,
  Nathan~A. Carr, Ondrej Stava, and Gavin S.~P. Miller.
\newblock Packmerger: {A} {3D} print volume optimizer.
\newblock {\em Comput. Graph. Forum}, 33(6):322--332, 2014.
\newblock \href {https://doi.org/10.1111/CGF.12353}
  {\path{doi:10.1111/CGF.12353}}.

\bibitem{worman2003decomposing}
Chris Worman.
\newblock Decomposing polygons into diameter bounded components.
\newblock In {\em Canadian Conference on Computational Geometry (CCCG)}, pages
  103--106, 2003.
\newblock URL: \url{http://www.cccg.ca/proceedings/2003/9.pdf}.

\bibitem{YAO20241002}
Yuan Yao, Longyu Cheng, and Zhengyu Li.
\newblock A comparative review of multi-axis {3D} printing.
\newblock {\em Journal of Manufacturing Processes}, 120:1002--1022, 2024.
\newblock \href {https://doi.org/10.1016/j.jmapro.2024.04.084}
  {\path{doi:10.1016/j.jmapro.2024.04.084}}.

\end{thebibliography}

\appendix

\section{Proof of \Cref{lem:nondegenerate}}
\label{app:A}

\begin{proof}
In order to obtain properly non-piercing pieces, we modify the boundaries of the pieces in $\mathcal Q$ slightly and eliminate all degeneracies.
We need to preserve that the modified pieces contain the same subsets of points in $X$ as the original pieces.
In general, the modifications we make to get from $\mathcal Q$ to $\mathcal Q'$ will ensure that the overlay of $\mathcal Q$ and $X$ is equivalent to the overlay of $\mathcal Q'$ and $X$ in the sense that cells in the former correspond to cells in the latter, and corresponding cells contain the same point from $X$.

Let $\mathcal U=\{U_1,\ldots,U_q\}$ be a set of unit squares such that $Q_i$ is a connected component of $P\cap U_i$, and let $\mathcal U'$ be defined similarly for the pieces $\mathcal Q'$ (which will be defined later).
There are two types of degeneracies in $\mathcal Q$ that we need to get rid of:
(i) overlapping edges of different pieces that stem from overlapping edges of squares in $\mathcal U$, and (ii) overlapping edges that stem from the boundary of $P$.
We first define a value $\delta$ which is how much we can modify the boundaries without changing the overlay substantially.
To this end, let $E$ be the combined set of edges in $\mathcal Q$ and $\mathcal U$ and let $V$ be the combined set of vertices in $\mathcal Q$ and $\mathcal U$ and the points $X$.
Let $\delta>0$ be the minimum distance from an edge $e\in E$ to a point in $V$ that is not on $e$.
We will move the boundaries of the pieces $\mathcal Q$ by a distance of at most $\delta/4$, which ensure that the pieces we construct contain the same subsets of points in $X$, and that we don't collapse any cells or introduce intersection points that don't correspond to intersection points in the overlay of $\mathcal Q$.

\begin{figure}[ht]
\centering
\includegraphics[page=3]{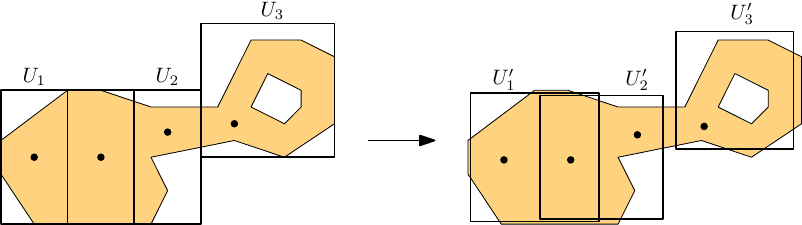}
\caption{The figure shows an example of how we shrink the squares slightly so that they all intersect properly.}
\label{fig:generalsquares}
\end{figure}

Let $U'_i$ be the axis-aligned square with side length $1-i\delta/2q$ and the same center as $U_i$; see \Cref{fig:generalsquares}.
This ensures that if squares $U_i$ and $U_j$ intersect in a non-proper way, then $U'_i$ and $U'_j$ are either disjoint or intersect in a proper way, while if $U_i$ and $U_j$ are disjoint or intersect properly, then $U'_i$ and $U'_j$ do the same.
Let $\tilde Q_i$ be the connected component of $P\cap U'_i$ that intersects $Q_i$; see \Cref{fig:generalpieces}.
Note that $\tilde Q_i$ may have holes that stem from holes of $P$, and according to the definition of non-piercing pieces, holes are not allowed.
We therefore define $\hat Q_i$ to be $\tilde Q_i$ where these holes have been filled.
Let $\hat {\mathcal Q}=\{\hat Q_1,\ldots,\hat Q_q\}$.

\begin{figure}[ht]
\centering
\includegraphics[page=4]{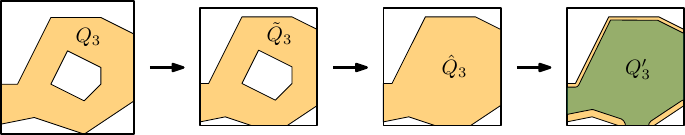}
\caption{The modifications we make to get from a piece $Q_3$ (from \Cref{fig:generalsquares}) to the piece $Q'_3$.
We shrink the defining square, we fill holes, and we offset the intervals from $\partial P$.}
\label{fig:generalpieces}
\end{figure}

In general, the boundaries of $\hat Q_i$ and $P$ share some intervals, which lead to non-proper intersections among the pieces $\hat {\mathcal Q}$.
Denote these intervals along $\partial\hat Q_i$ by $\Gamma_i$, i.e., $\Gamma_i$ contains the connected components of $\partial P\cap\partial\hat Q_i$.
Let $\Gamma=\bigcup_{1\leq i\leq q}\Gamma_i$ be the intervals of all pieces.
Let $\kappa>0$ be the minimum difference in lengths of two intervals in $\Gamma$ that have different lengths (i.e., the minimum positive gap in the sorted order of the lengths of $\Gamma$).
We get $Q'_i$ from $\hat Q_i$ by offsetting each shared interval $\gamma\in\Gamma_i$ inwards by a distance of
\[x=\frac{\delta}{4(1+\vert \gamma\vert+\kappa i/2q)},
\]
where $\vert \gamma\vert$ denotes the length of $\gamma$.
In other words, we remove the open region $\gamma\oplus D(x)$ from $\hat Q_i$, where $\oplus$ denotes the Minkowski sum and $D(x)$ is the open disk with radius $x$ centered at the origin.
Note that the offset is always $x<\delta/4$.
This ensures that the overlay of $\hat{\mathcal Q}\cup \mathcal U'\cup X$ is combinatorially equivalent to that of $\mathcal Q'\cup \mathcal U'\cup X$.
For each curve $\gamma\in\Gamma$ on the boundary of a piece $Q_i$, there is a corresponding curve $\gamma'$ on the boundary of the piece $Q'_i$.
Another important property is that we remove less along long intervals in $\Gamma$ than along short ones; this idea comes from~\cite{DBLP:journals/algorithmica/BergKMT23}.
The term $\kappa i/2q$ is used to break ties among intervals of the same length from different pieces, which ensures that we avoid introducing new improper intersections by offsetting two different intervals by exactly the same amount.
We conclude that the pieces $\mathcal Q'=\{Q'_1,\ldots,Q'_q\}$ intersect properly.
Furthermore, we have ensured that each piece $Q'_i$ contains the same points in $X$ as the original $Q_i$, and that the overlays of $\mathcal Q$ and $\mathcal Q'$ are equivalent.

It remains to verify that the pieces $\mathcal Q'$ are non-piercing.
To this end, consider two pieces $Q'_i$ and $Q'_j$, and we will verify that $Q'_i\setminus Q'_j$ is connected.
If $Q'_i\cap Q'_j=\emptyset$, the claim follows trivially as $Q'_i$ is connected, so assume that $Q'_i$ and $Q'_j$ overlap.
Then their associated squares $U'_i$ and $U'_j$ also overlap, and we define $B=U'_i\cap U'_j$ to be the rectangular intersection.
Let $\beta$ be the edges of $B$ that are contained in the edges of $U'_j$ (so $\beta$ is a curve that passes through the interior of $U'_i$).

We first consider the special case $\hat Q_i\subset B$.
Here we have $\hat Q_i\subset\hat Q_j$ (since we assume $Q'_i\cap Q'_j\neq\emptyset$).
Note that all intervals in $\Gamma_i$ are subsets of longer ones in $\Gamma_j$, so we offset more when constructing $Q'_i$ than when constructing $Q'_j$.
It thus follows that $Q'_i\subset Q'_j$, and hence that $Q'_i\setminus Q'_j=\emptyset$, which is a connected set.

On the other hand, if $\hat Q_i\setminus B\neq\emptyset$, there is also a point $s\in Q'_i\setminus B$ (since we preserve cells in the overlay of $\mathcal Q$).
Consider another point $t\in Q'_i$, and we will verify that $s$ and $t$ are connected by a path in $Q'_i\setminus Q'_j$.
We will in fact describe a path in the closure $\overline{Q'_i\setminus Q'_j}$, and this path may have intervals along $\beta$ that are strictly speaking not in $Q'_i\setminus Q'_j$, but these intervals can be moved out of $B$ by an infinitesimal amount to obtain a path in $Q'_i\setminus Q'_j$.
Let $\pi$ be a path between $s$ and $t$ in $Q'_i$.
Assume first that $t\in Q'_i\setminus B$.
Note that the intervals of $\pi$ that are not in $Q'_i\setminus Q'_j$ must be in $B$ with endpoints on $\beta$.
Consider a connected component $\pi_1$ of $\pi\cap Q'_j$, with endpoints $a,b\in \beta$.

\begin{figure}[ht]
\centering
\includegraphics[page=5,width=\textwidth]{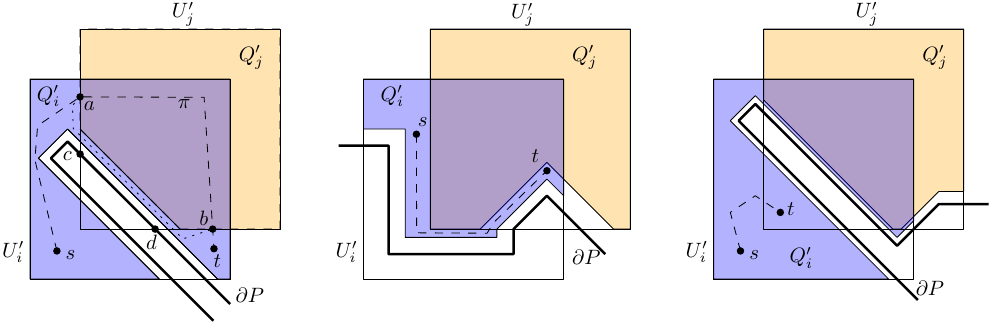}
\caption{Left: The path $\pi$ (dashed curve) connecting $s$ and $t$ passes through $Q'_j$ between $a$ and $b$.
In order to reach $t$ and stay in $Q'_i\setminus Q'_j$, we can instead pass through the corridor created between the two offset curves stemming from the interval $\partial P[c,d]$, shown as a dotted curve.
Middle: If $t\in B$ and $t\in \hat Q_j$, then $t$ must be in a corridor between two offsets of an interval of $\partial P$.
We can leave $B$ through the corridor, which reduces the case to that where $t\notin B$.
Right: If $t\in B$ but $t\notin \hat Q_j$, then we can readily reach the part of $Q'_i$ outside $B$ and thus reduce to the case where $t$ is not in $B$.}
\label{fig:corridor}
\end{figure}

We show that we can replace $\pi_1$ by another curve $\pi_2$ with the same endpoints $a,b$, so that $\pi_2$ stays in $\overline{Q'_i\setminus Q'_j}$; see \Cref{fig:corridor} (left).
We consider traversing $\beta$ from $a$ towards $b$.
If we leave $Q'_i$, it is because we cross the offset curve $\gamma'$ corresponding to some curve $\gamma\in\Gamma_i$.
Hence, there is an interval $\partial P[c,d]\subset\gamma$ with $c,d\in\beta[a,b]$ that is contained in the region enclosed by $\pi_1\cup \beta[a,b]$.
The interval $\partial P[c,d]$ is on the boundary of both $\hat Q_i$ and $\hat Q_j$.
We note that $\partial P[c,d]$ is a curve in $\Gamma_j$ which is shorter than $\gamma$, since $\gamma$ extends from both $c$ and $d$ into the interior of $U'_i$.
We therefore make a greater offset to get $Q'_j$ from $\hat Q_j$ than to get $Q'_i$ from $\hat Q_i$.
Hence, the curve $\gamma'$ is bounding $Q'_i\setminus Q'_j$, so we can follow $\gamma'$ through $B$ and get to a point on $\beta$ closer to $b$.

On the other hand, suppose that $t\in B$.
We again have two cases, namely whether $t\in \hat Q_j$ or not.
If $t\in \hat Q_j$, note that since $t\in\hat Q_i\cap\hat Q_j$, the reason that $t\notin Q'_j$ must be that $t$ is in the corridor between two offsets of intervals in $\Gamma_i$ and $\Gamma_j$, respectively, so the offset in $\hat Q_j$ removed $t$, but the offset in $\hat Q_i$ preserved $t$; see \Cref{fig:corridor} (middle).
Let the corresponding intervals in $\Gamma_i$ and $\Gamma_j$ be $\gamma_i$ and $\gamma_j$, respectively.
Then $\gamma_i$ and $\gamma_j$ overlap in $B$, and we know that $\gamma_i$ is longer than $\gamma_j$.
This is only possible if $\gamma_j$ has an endpoint on $\beta$ from which $\gamma_i$ continues into the interior of $U'_i\setminus B$.
Therefore, we can follow the offset path $\gamma'_i$ from $b$ until we leave $B$.
The proof then reduces to the case where $t$ is outside~$B$.

Finally, suppose that $t\notin \hat Q_j$ (but still $t\in B$); see \Cref{fig:corridor} (right).
Then $t$ is in another connected component of $P\cap U'_j$ than $\hat Q_j$.
Hence, we can take any path from $t$ in $Q'_i$ to get outside $B$ and then reduce to the case where $t$ is outside~$B$.
\end{proof}

\section{Proof of \Cref{thm:sparse-reduction}}\label{sec:bigpoly}

\subsection{Constructing a set of simple instances}\label{sec:simpleInstances}
In this section, we proceed to describe the set of corridors $\mathcal{C}$.
We start with the following definition of corridors, see also~\cref{fig:rhombus}.

\begin{figure}[ht]
    \centering
    \includegraphics{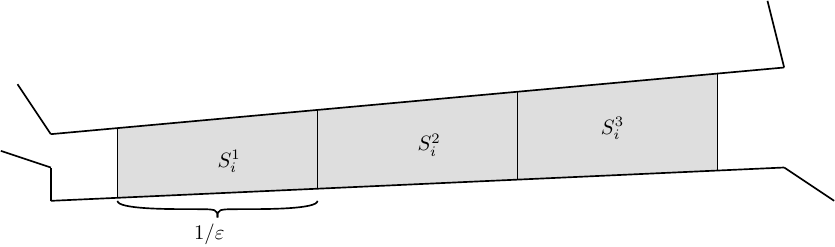}
    \caption{A corridor $C_i$ in gray with two vertical sides. The corridor $C_i$ is further partitioned into sections $S_i^j$ with two vertical sides all having length exactly $1/\varepsilon$.}
    \label{fig:rhombus}
\end{figure}

Let $C,P$ be two polygons with $C\subset P$. We say that $C$ is a \emph{horizontal corridor}, if $C$ has four vertices, two vertical edges, and its two other edges are contained in edges of~$P$. We define the \emph{length} of $C$ as the distance between its two vertical edges. We say that $C$ is \emph{maximal}, if $C$ is not strictly contained in any other horizontal corridor $C'$. (Maximal) vertical corridors are defined similarly.

We start with a few useful lemmas on maximal corridors.
\begin{lemma} \label{lem:disjointCorridors}
Let $C,C'$ be two maximal horizontal corridors. Then either $C=C'$ or $C$ and $C'$ have disjoint interiors. The same claim holds for maximal vertical corridors.
\end{lemma}
\begin{proof}
It is to check that if the interiors of $C$ and $C'$ intersect, then the non-horizontal edges of $C$ and $C'$ must be contained in the same edges of $P$. But this implies that $C\cup C'$ is also a horizontal corridor. By maximality, $C=C'$.
\end{proof}
Let $\mathcal{M}_h$ denote the set of maximal horizontal corridors of $P$. Then, by \cref{lem:disjointCorridors}, these corridors are pairwise interior disjoint. We next want to bound the size of $\mathcal{M}_h$.  We note that it is easy to bound the size of $\mathcal{M}_h$ by $O(n^2)$ which would suffice for our purposes. However, lending ideas from~\cite{aamand2023tiling}, we can give a tighter bound. 
\begin{lemma}\label{lemma:size}
If $P$ has $n$ corners, then $|\mathcal{M}_h|\leq 3n-5$. 
\end{lemma}
\begin{proof}
Similarly to the proof of Lemma 5.4 of~\cite{aamand2023tiling}, we define a multigraph $G=(V,E)$ where $V$ is the set of  edges of $P$  and for each horizontal corridor $C\in \mathcal{M}_h$, we add an edge in $E$ between the two edges of $P$ that bound it. By the disjointness of maximal corridors guaranteed by \cref{lem:disjointCorridors}, this $n$-vertex graph has a planar embedding $\mathcal{E}$ with the property that for any two parallel edges $e,e'$ with $e\neq e'$, the Jordan curve formed by $e$ and $e'$ in $\mathcal{E}$ contains a vertex of $V$ in its interior. Lemma 5.3 of~\cite{aamand2023tiling} bounds the number of edges of such a graph by $3n-5$ when $n\geq 2$.
\end{proof}
Now define $\mathcal{M}_h'\subset \mathcal{M}_h$ to be the maximal corridors of length at least $1/\eps+4$. For a fixed $C\in \mathcal{M}_h'$, let $t\in \mathbb{N}$ be such that the length $\ell$ of $C$ satisfies $t\cdot 1/\eps+4\leq \ell< (t+1)\cdot 1/\eps+4$. Let $C_0\subset C$ be a horizontal corridor of length $t\cdot 1/\eps$ centered in the middle of $C$. We form a set $\mathcal{C}_h$ consisting of all corridors $C_0$ that can be obtained in the above fashion starting with a corridor $C\in \mathcal{M}_h'$.

Next we define $P'$ to be $\overline{P\setminus \bigcup \mathcal{C}_h}$. {Note that $P'$ could be disconnected; each component of $P'$ is a polygon.}
We have the following lemma.
\begin{lemma}\label{lemma:vertical-width}
Any connected component of $P'$ is contained in a vertical strip of width $O(n/\eps)$.
\end{lemma}
\begin{proof}
Let $T$ be a connected component of $P'$. By~\Cref{lemma:size}, and the fact that removing a corridor $C\in \mathcal C_h$ from $P$ creates at most $4$ new vertices, $T$ has at most $13 n$ vertices. Moreover, clearly $T$ has no maximal horizontal corridors of length $\geq 1/\eps+4$. In particular, projecting the vertices of $P'$ onto the $x$-axis, the distance between any two consecutive points in the projection is at most $1/\eps+4$. Since $T$ is connected, the claim follows.
\end{proof}
Analogous to $\mathcal{C}_h$, construct the set $\mathcal{C}_v$ starting with maximal vertical corridors of $P'$. By~\cref{lemma:size}, $P'$ has $O(n)$ vertices, so it follows by a similar argument that $|\mathcal{C}_v|=O(n)$.
We define $\mathcal{C}=\mathcal{C}_h\cup \mathcal{C}_v$.
Finally, we define $Q=\overline{P\setminus \bigcup \mathcal{C}}$.
By a similar argument as for~\cref{lemma:vertical-width}, we can show the following lemma.
\begin{lemma}\label{lemma:diameter}
    Any connected component of $Q$ is contained in an axis-aligned square of side length $O(n/\eps)$ and in particular has diameter $O(n/\eps)$.
\end{lemma}

This concludes the description of $\mathcal{C}$. It is easy to determine $\mathcal{C}$ in $\poly(n)$ time.

\subsection{Approximating the separate instances suffices}
We now show that we can solve the instances in $\mathcal{C}$ as well as $Q$ separately without affecting the approximation guarantee too much.
The following lemma will be useful.
\begin{lemma}\label{lemma:monotonicity}
Let $S,T$ be convex polygons with $S\subset T$. Let $\OPT_S$ and $\OPT_T$ be the minimum number of small pieces needed to cover $S$ and $T$. Then $\OPT_S\leq \OPT_T$.
\end{lemma}
\begin{proof}
Start with an optimal cover $\mathcal{C}$ of $T$. By convexity of $T$, we may assume that each piece of $\mathcal{C}$ is the intersection of a unit square and $T$.
Since intersections preserve convexity, each piece of $\mathcal{C}$ is convex.
Now by convexity of $S$, intersecting each piece of $\mathcal{C}$ with $S$ gives a collection of convex and hence connected pieces covering $S$.
The result follows.
\end{proof}
We prove that it suffices to solve the separate instances in $\mathcal{C}$ and $Q$ in two steps.
\begin{lemma}\label{lemma:separate-horizontal}
    Let $\mathcal{C}_h=\{C_1,\dots,C_{k'}\}$. Denote by $\OPT_i$ the size of an optimal cover of $C_i$ and by $\OPT_{P'}$ the size of an optimal cover of $P'$.
    Then 
    \[
    \OPT\leq \OPT_{P'} +\sum_{i=1}^{k'} \OPT_i\leq(1+O(\eps))\OPT.
    \]
\end{lemma}
\begin{proof}
The first inequality is clear. For the second inequality, we need to consider the boundary regions between the $C_i$ and $P'$. Fix some $C_i\in \mathcal{C}_h$. Let $\ell_i^{(1)}$ and $\ell_i^{(2)}$ be the length of the two vertical edges $e_1$ and $e_2$ of $C_i$. 

We first show that in the optimal cover of $P$, at most $O(\ell_i^{(1)}+\ell_i^{(2)}+1)$ pieces cross these vertical edges. 
To see this, let $B_1\subset C_i$ be the set of points of $\ell_1$-distance at most $1$ to the left vertical edge $e_1$ of $C_i$ (of length $\ell_i^{(1)}$). By the way we defined $C_i$ (by cutting {off} corridors of length at least $2$ from either end of a maximal corridor), any piece of the optimal solution intersecting $e_1$ is contained in $B_1$. It is easy to check that $B_1$ can be covered with {$O(\ell_i^{(1)} + 1)$} pieces, see also \cref{fig:ends}. A similar argument holds for $e_2$ and the claim follows.

\begin{figure}[ht]
    \centering
    \includegraphics[page=2]{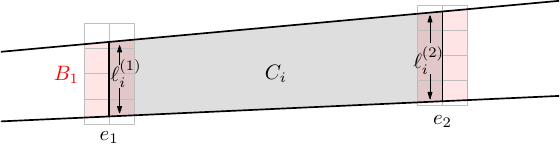}
    \caption{Illustration for the proof of \cref{lemma:separate-horizontal}: At most $O(\ell_i^{(1)}+\ell_i^{(2)}+1)$ pieces cross $e_1$ or $e_2$.}
    \label{fig:ends}
\end{figure}

This implies that from a cover of $P$ of size $\OPT$, we can obtain separate covers of the polygons in $\mathcal{C}_h$ and $P'$ by increasing the size of the cover by at most $O\left(\sum_{i=1}^{k'}(\ell_i^{(1)}+\ell_i^{(2)}+1)\right)$. Since all $C_i\in \mathcal{C}_h$ have length at least $1/\eps$, by a volume argument, an optimal cover of $P$ must use at least $\Omega(\sum_{i=1}^{k'}(\ell_i^{(1)}+\ell_i^{(2)}+1)/\eps)$ pieces. It follows that
\[
\OPT_{P'} +\sum_{i=1}^{k'} \OPT_i\leq \OPT+O\left(\sum_{i=1}^{k'}(\ell_i^{(1)}+\ell_i^{(2)}+1)\right)\leq (1+O(\eps)) \OPT.
\]
This concludes the proof.
\end{proof}
By a similar argument, we get the following.
\begin{lemma}\label{lemma:separate-vertical}
    Let $\mathcal{C}=\{C_1,\dots,C_k\}$. Denote by $\OPT_i$, the size of an optimal cover of $C_i$ and $\OPT_{Q}$, the size of an optimal cover of $Q$. Then 
    \[
    \OPT\leq \OPT_{Q} +\sum_{i=1}^k \OPT_i\leq(1+O(\eps))\OPT.
    \]
\end{lemma}
\begin{proof}
Let $\mathcal{C}_v=\{C_{k'+1},\dots,C_k\}$.
Denote by $\OPT_i$ the size of an optimal cover of $C_i$ and by $\OPT_Q$ the size of an optimal cover of $Q$.
Identically to above, we can prove that 
\[
    \OPT_{P'}\leq \OPT_{Q} +\sum_{i=k'+1}^{k} \OPT_i\leq(1+O(\eps))\OPT_{P'},
    \]
The result now follows from the bound of~\cref{lemma:separate-horizontal}.
\end{proof}
From~\cref{lemma:separate-vertical}, it also follows that if we can find an $1+O(\eps)$ approximation to each $C_i\in \mathcal{C}$ and for $Q$, we obtain an $1+O(\eps)$ approximation for $P$. We set out to do this next.

\subsection{Approximating the separate instances}

We now prove the main theorem of the section.
\largeDiameter*

\begin{proof}[Proof of~\cref{thm:sparse-reduction}]
The first step of $\mathcal{A}'$ is to identify the set of corridors $\mathcal{C}$ and the remainder $Q=\overline{P\setminus \bigcup \mathcal{C}}$ as described in \cref{sec:simpleInstances}. This can clearly be done in $\poly(n)$ time.
By~\cref{lemma:separate-vertical}, it suffices to determine $(1+O(\eps))$-approximate solutions for $Q$ and all corridors in $\mathcal{C}$ separately.  By~\cref{lemma:diameter}, any connected component $Q$ has diameter $O(n/\eps)$ and there are at most $O(n)$ of them. Thus, these instances can be handled by $\mathcal A$ in $f(n,\eps)$ time. Next consider $C\in \mathcal{C}$; without loss of generality, we assume that $C$ is a horizontal corridor. Let $\ell_1$ be the length of its left vertical edge and $\ell_2$ be the length of its right vertical edge. Without loss of generality, we may assume $\ell_1\leq \ell_2$. 
\begin{figure}[ht]
    \centering
    \includegraphics[width=\textwidth]{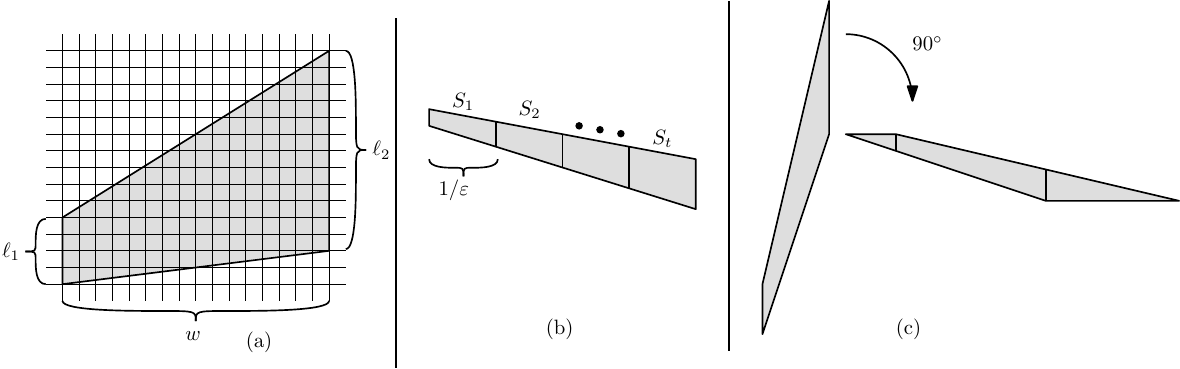}
    \caption{Three of the four cases for $C_i$ (Case D is similar to Case A). (a) One of the non-vertical edges has slope in $[-2,2]$ and $\ell_2\geq 1/\eps$. (b) One of the non-vertical edges has slope in $[-2,2]$ and $\ell_2\leq 1/\eps$. (c) Both slopes of non-vertical edges are in $[2,\infty)$ or both slopes of non-vertical edges are in $(-\infty,-2]$. In case (a), imposing an {integer} grid suffices to define an optimal cover. In case (b), we will consider a further subdivision of $C$ into corridors of length $1/\eps$ and show that it suffices to query $\mathcal{A}$ on $\poly(1/\eps)$ of them. Case (c) reduces to case (a) or (b) after a $90^\circ$ rotation.}
    \label{fig:cases}
\end{figure}
We consider three cases, see also \cref{fig:cases}.
\begin{itemize}
\item \textbf{Case A} One of the non-vertical edges has slope in $[-2,2]$ and $\ell_2\geq 1/\eps$.
\item \textbf{Case B} One of the non-vertical edges has slope in $[-2,2]$ and $\ell_2\leq 1/\eps$.
\item \textbf{Case C} Both slopes of non-vertical edges are in $[2,\infty)$ or both are in $(-\infty,-2]$.
\item \textbf{Case D} One non-vertical edge has slope in $(-\infty,-2]$ and the other has slope in $[2,\infty)$.
\end{itemize}
We handle these cases separately.

\textbf{Case A} Let $w$ be the length of $C$. Then the area $A$ of $C$ is $\Omega(\ell_2\cdot w)$, so any cover must use at least that many small pieces. By the assumption on the slopes, at most $O(\ell_2+w)$ cells of the {integer} grid intersect the boundary of $C$. In particular, any unit grid defines a covering of $R$ using at most $A+O(\ell_2+w)=A+O(\max(\ell_2,w))$ small pieces. Since both $\ell_2\geq 1/\eps$ and $w\geq 1/\eps$, we have $\max(\ell_2,w)=\ell_2w/\min(\ell_2,w) \leq \ell_2w\eps\in O(A\eps)$. It follows that the covering uses at most $A(1+O(\eps))\leq \OPT_R(1+O(\eps))$ pieces. We can thus output the estimator $\lambda_R=A+\alpha\cdot \eps$ to get that $\OPT_R\leq \lambda_R\leq (1+O(\eps))R$ for some constant $\alpha$. 

\textbf{Case B} By construction, the corridor $C$ has length $t\cdot 1/\eps$ for some $t\in \mathbb{N}$, so we can partition $C$ into $t$ corridor sections $S_1,\dots,S_t$ of length $1/\eps$ overlapping only at their boundary. 
By an identical argument to~\cref{lemma:separate-horizontal}, if $\OPT_i$ and $\OPT_C$ denote the size of an optimal cover of $S_i$ and $C$, respectively, then $\sum_{i=1}^t\OPT_i\leq (1+O(\eps))\OPT_C$.
We will not always be able to approximate $\OPT_i$  for all $i$. However, we obtain such estimates for all but at most a $1/\poly(1/\eps)$ fraction of the sections and we will show that this suffices. 

If $t\leq 2/\eps^4$, we approximate all $\OPT_i$ with a factor in $(1+O(\eps))$ by running $\mathcal{A}$ on each of the $t$ instances. 

So assume $t>2/\eps^4$.
Let $\ell\in [t\eps^4,2t\eps^4]$ be an integer and $J'=\{\ell,2\ell,\dots, \lfloor t/\ell\rfloor\ell\}$ and $J=J'\cup \{t\}$.
Our algorithm $\mathcal{A}'$ runs $\mathcal{A}$ on each of the instances $(S_j)_{j\in J}$. Denote $\lambda_j$ the estimate returned by $\mathcal{A}$. Our estimator $\lambda_C$ for $\OPT_C$ is simply
\[
\lambda_C=\sum_{j\in J}\ell\lambda_j.
\]
We now show that it provides the desired approximation. Since $\lambda_j\geq \OPT_j$ and $\OPT_j$ is non-decreasing in $j$, it follows that for $j\in J$, $\ell\lambda_j$ is an upper bound on all of $\sum_{i=0}^{\ell-1}\OPT_{j-i}$. Since for any $j'\in [t]$, there exists  $j\in J$, such that the corresponding sums includes a term $\OPT_{j'}$, it follows directly that 
\[
\lambda_C\geq \sum_{i=1}^t \OPT_i\geq \OPT_C.
\]
To provide a corresponding upper bound write $J=\{j_1,\dots,j_k\}$ in increasing order where $k=O(1/\eps^4)$.
By~\Cref{lemma:monotonicity} and the approximation guarantee of $\mathcal{A}$, for any $a<b$,
\begin{align}\label{eq:monotone-estimates}
\lambda_{j_{a}}\leq (1+\eps)\OPT_{j_a}\leq (1+\eps)\OPT_{j_b}\leq (1+\eps)\lambda_{j_b} 
\end{align}
Since $\OPT_i=\Omega(1/\eps)$ but $\OPT_i=O(1/\eps^2)$ for all $i$, by the approximation guarantee of $\mathcal{A}$, it also holds that $\lambda_{j_a}=\Omega(1/\eps)$ and $\lambda_{j_a}=O(1/\eps^2)$ for all $a$. 
This implies that there can be at most $O(\log(1/\eps)/\eps)$  values $a$ such that $\lambda_{j_a}>(1+2\eps)\lambda_{j_{a-1}}$.
To see this, note that after such a jump happens, by~\cref{eq:monotone-estimates}, all $j_b$ with $b>a$ must have $j_b>j_{a-1}(1+2\eps)/(1+\eps)$ and so, if $m$ is the number of such jumps, then
\[
\left(\frac{1+2\eps}{1+\eps}\right)^m=O(1/\eps),
\]
which implies that $m=O(\log(1/\eps)/\eps)$.

Denote the set of such $j_a$ by $J_0\subset J$. We also include $t=j_k$ in the set $J_0$. Now if $a\in J\setminus J_0$, then $\lambda_{j_{a}}\leq (1+2\eps)\lambda_{j_{a-1}}$, but this implies that for any $j$ with $j_{a-1}<j\leq j_a$,
\[
\lambda_{j_a}\leq (1+2\eps)\lambda_{j_{a-1}}\leq (1+2\eps)(1+\eps) \OPT_{j_{a-1}}\leq (1+O(\eps)) \OPT_{j}.
\]
It follows that 
\[
\sum_{j\in J\setminus J_0}\ell\lambda_j\leq (1+O(\eps)) \sum_{j\in J\setminus J_0}\sum_{i=j-\ell+1}^{j} \OPT_i\leq (1+O(\eps))\sum_{i=1}^t \OPT_i=(1+O(\eps))\OPT_C.
\]
Moreover, note that $\OPT_C=\Omega(t/\eps)$ by a diameter argument. Since additionally, all $\lambda_j=O(1/\eps^2)$ and $J_0=O(\log (1/\eps)/\eps)$, it follows that 
\[
\sum_{j\in J_0}\ell\lambda_j=O(t\eps^4\cdot 1/\eps^2\cdot \log (1/\eps)/\eps)=O( t)=O(\eps)\cdot \OPT_C
\]
It follows that our estimator satisfies $\lambda_C\leq (1+O(\eps))\OPT_C$, as desired.

\textbf{Case C} In this case, we reduce to Case A or B by rotating $R$ by $90^\circ$ (importantly, a unit square rotated around its center by $90^\circ$ is the same unit square) and partitioning it into three polygons as in~\cref{fig:cases} (c). As in the proof of~\cref{lemma:separate-vertical}, one can show that it suffices to find $1+O(\eps)$ approximations for each of these pieces separately and it is simple to check that these pieces either fall into Case A or B or have diameter $O(1/\eps)$ in which case we can find $(1+\eps)$-approximate solutions directly.

\textbf{Case D} The argument is similar to Case A. Note that the assumption on the slopes ensures that $\ell_2\geq 1/\eps$.

\subparagraph{Running time} We now bound the runtime of our algorithm. Finding $\mathcal{C}$ and $Q$  takes $\poly(n)$ time. Since $Q$ has at most $\alpha n$ vertices and diameter $\alpha n/\eps$, for some large constant $\alpha$, we can find a $(1+\eps)$-approximate solution for $Q$ in time $f(\alpha n,\eps)$. For a corridor $C\in \mathcal{C}$ of Case A or D, our estimator of $\OPT_C$ just requires us to compute the area of $A$ which can be done in $O(1)$ time. For Case B, for each of the $O(n)$ corridors $C\in \mathcal{C}$, we run $\mathcal{A}$ on $O(1/\eps^4)$ instances each taking time $f(4,\eps)$ resulting in a total runtime of $nf(4,\eps)/\eps^4$. 
As Case C reduces to constantly many cases of Case A or Case B, we can bound the running time similarly. Combining the bounds give the desired running time of $f(\alpha n,\eps)\cdot \poly(n/\eps)$.

\subparagraph{Implicit description}
For $Q$, the algorithm can maintain an exact description of the approximate cover, so for $q\in Q$, it is easy to find the pieces containing $q$ in $O(\poly (n/\eps))$ time.  For $C\in \mathcal{C}$ of the type in Case A or D, we simply need to check which grid cell of the uniform grid $q$ lies. This  takes $O(1)$ time after identifying $C$ which takes $\poly(n)$ time. In Case B, we first identify the section $S_i$ which contains $q$ and then run $\mathcal{A}$ on $S_i$. This takes $\poly(n)+f(4,\eps)$ time. Finally, Case C reduces to Case A or B.
\end{proof}

\end{document}